\documentclass[onecolumn,10pt]{IEEEtran}
\usepackage[mathscr]{eucal}
\usepackage[cmex10]{amsmath}
\usepackage{epsfig,epsf,psfrag}
\usepackage{amssymb,amsmath,amsthm,amsfonts,latexsym}
\usepackage{amsmath,graphicx,bm,xcolor,url,overpic}
\usepackage{fixltx2e}
\usepackage{array}
\usepackage{verbatim}
\usepackage{bm}
\usepackage{algorithmic}
\usepackage{algorithm}
\usepackage{verbatim}
\usepackage{textcomp}
\usepackage{mathrsfs}
\usepackage{epstopdf}

\newcommand{\openone}{\leavevmode\hbox{\small1\normalsize\kern-.33em1}}

\catcode`~=11 \def\UrlSpecials{\do\~{\kern -.15em\lower .7ex\hbox{~}\kern .04em}} \catcode`~=13 

\allowdisplaybreaks[4]

\newcommand{\nn}{\nonumber}

\newcommand{\calA}{\mathcal{A}}

\newcommand{\calF}{\mathcal{F}}

\newcommand{\calP}{\mathcal{P}}

\newcommand{\calS}{\mathcal{S}}
\newcommand{\calT}{\mathcal{T}}

\newcommand{\calX}{\mathcal{X}}
\newcommand{\calY}{\mathcal{Y}}

\newcommand{\ba}{\mathbf{a}}

\newcommand{\bs}{\mathbf{s}}
\newcommand{\bS}{\mathbf{S}}

\newcommand{\bw}{\mathbf{w}}
\newcommand{\bW}{\mathbf{W}}
\newcommand{\bx}{\mathbf{x}}
\newcommand{\bX}{\mathbf{X}}

\newcommand{\bY}{\mathbf{Y}}

\newcommand{\rma}{\mathrm{a}}

\newcommand{\rmb}{\mathrm{b}}

\newcommand{\rmd}{\mathrm{d}}

\newcommand{\rme}{\mathrm{e}}

\newcommand{\rmG}{\mathrm{G}}

\newcommand{\rmP}{\mathrm{P}}
\newcommand{\rmq}{\mathrm{q}}

\newcommand{\rmV}{\mathrm{V}}

\newcommand{\bbN}{\mathbb{N}}

\newcommand{\bbR}{\mathbb{R}}

\DeclareMathAlphabet{\mathbsf}{OT1}{cmss}{bx}{n}
\DeclareMathAlphabet{\mathssf}{OT1}{cmss}{m}{sl}

\DeclareSymbolFont{bsfletters}{OT1}{cmss}{bx}{n}  
\DeclareSymbolFont{ssfletters}{OT1}{cmss}{m}{n}
\DeclareMathSymbol{\bsfGamma}{0}{bsfletters}{'000}
\DeclareMathSymbol{\ssfGamma}{0}{ssfletters}{'000}
\DeclareMathSymbol{\bsfDelta}{0}{bsfletters}{'001}
\DeclareMathSymbol{\ssfDelta}{0}{ssfletters}{'001}
\DeclareMathSymbol{\bsfTheta}{0}{bsfletters}{'002}
\DeclareMathSymbol{\ssfTheta}{0}{ssfletters}{'002}
\DeclareMathSymbol{\bsfLambda}{0}{bsfletters}{'003}
\DeclareMathSymbol{\ssfLambda}{0}{ssfletters}{'003}
\DeclareMathSymbol{\bsfXi}{0}{bsfletters}{'004}
\DeclareMathSymbol{\ssfXi}{0}{ssfletters}{'004}
\DeclareMathSymbol{\bsfPi}{0}{bsfletters}{'005}
\DeclareMathSymbol{\ssfPi}{0}{ssfletters}{'005}
\DeclareMathSymbol{\bsfSigma}{0}{bsfletters}{'006}
\DeclareMathSymbol{\ssfSigma}{0}{ssfletters}{'006}
\DeclareMathSymbol{\bsfUpsilon}{0}{bsfletters}{'007}
\DeclareMathSymbol{\ssfUpsilon}{0}{ssfletters}{'007}
\DeclareMathSymbol{\bsfPhi}{0}{bsfletters}{'010}
\DeclareMathSymbol{\ssfPhi}{0}{ssfletters}{'010}
\DeclareMathSymbol{\bsfPsi}{0}{bsfletters}{'011}
\DeclareMathSymbol{\ssfPsi}{0}{ssfletters}{'011}
\DeclareMathSymbol{\bsfOmega}{0}{bsfletters}{'012}
\DeclareMathSymbol{\ssfOmega}{0}{ssfletters}{'012}

\newcommand{\tili}{\tilde{i}}

\newcommand{\tilM}{\tilde{M}}

\newcommand{\tilP}{\tilde{P}}

\newcommand{\hats}{\hat{s}}
\newcommand{\hatS}{\hat{S}}

\newcommand{\hatt}{\hat{t}}

\newcommand{\hatw}{\hat{w}}
\newcommand{\hatW}{\hat{W}}

\newcommand{\barx}{\bar{x}}

\newcommand{\barP}{\bar{P}}

\newcommand{\barX}{\bar{X}}

\DeclareMathOperator*{\argmax}{arg\,max}

\newtheorem{theorem}{Theorem}

\newtheorem{corollary}{Corollary}

\newtheorem{definition}{Definition}

\usepackage[utf8]{inputenc} 
\usepackage[T1]{fontenc}    
\usepackage{url,bbm}            
\usepackage{booktabs}       
\usepackage{amsfonts,algorithm,algorithmic}       
\usepackage{nicefrac}       
\usepackage{microtype}      

\usepackage{cite}
\usepackage{subfigure,transparent,color,graphicx} 

\graphicspath{{./figures2/}}
\allowdisplaybreaks[2]
\newcommand{\bbo}{\mathbbm{1}}

\begin{document}

\title{Resolution Limits for the Noisy Non-Adaptive 20 Questions Problem}
\author{Lin Zhou and Alfred Hero \\

\thanks{A preliminary version of this paper was presented at ISIT 2020. This work was supported in part by the National Key Research and Development Program of China under Grant 2020YFB1804800 and in part by ARO grant W911NF-15-1-0479.}
\thanks{Lin Zhou is with the School of Cyber Science and Technology, Beihang University, Beijing 100191, China, and also with the Beijing Laboratory for General Aviation Technology, Beihang University, Beijing 100191, China (Email: lzhou@buaa.edu.cn). He was with the Department of Electrical Engineering and Computer Science, University of Michigan, Ann Arbor, MI, USA, 48109-2122.}
\thanks{Alfred Hero is with the Department of Electrical Engineering and Computer Science, University of Michigan, Ann Arbor, MI, USA, 48109-2122 (Email: hero@eecs.umich.edu.).}
}
\maketitle

\begin{abstract}
We establish fundamental limits on estimation accuracy for the noisy 20 questions problem with measurement-dependent noise and introduce optimal non-adaptive procedures that achieve these limits. The minimal achievable resolution is defined as the absolute difference between the estimated and the true locations of a target over a unit cube, given a finite number of queries constrained by the excess-resolution probability. Inspired by the relationship between the 20 questions problem and the channel coding problem, we derive non-asymptotic bounds on the minimal achievable resolution to estimate the target location. Furthermore, applying the Berry--Esseen theorem to our non-asymptotic bounds, we obtain a second-order asymptotic approximation to the achievable resolution of optimal non-adaptive query procedures with a finite number of queries subject to the excess-resolution probability constraint. We specialize our second-order results to measurement-dependent versions of several channel models including the binary symmetric, the binary erasure and the binary Z- channels. As a complement, we establish a second-order asymptotic achievability bound for adaptive querying and use this to bound the benefit of adaptive querying.
\end{abstract}

\begin{IEEEkeywords}
20 Questions, Resolution, Non-adaptive, Adaptive, Second-order asymptotics, Finite blocklength analysis, Multidimensional target, Sorted posterior matching, Probably approximately correct learning
\end{IEEEkeywords}

\section{Introduction}
The noisy 20 questions problem (cf. \cite{renyi1961problem,burnashev1974interval,ulam1991adventures,pelc2002searching,jedynak2012twenty,chung2018unequal,lalitha2018improved}) arises when one aims to accurately estimate an arbitrarily distributed random variable $S$ by successively querying an oracle and using noisy responses to form an estimate $\hatS$. A central goal in this problem is to find optimal query strategies that yield a good estimate $\hatS$ of the unknown target $S$.

Depending on the query design framework, the 20 questions problem can either be adaptive or non-adaptive. In adaptive query procedures, the design of a subsequent query depends on all previous queries and noisy responses to these queries from the oracle. In non-adaptive query procedures, all the queries are designed independently in advance. For example, the bisection policy~\cite[Section 4.1]{jedynak2012twenty} is an adaptive query procedure and the dyadic policy~\cite[Section 4.2]{jedynak2012twenty} is a non-adaptive query procedure. Compared with adaptive query procedures, non-adaptive query procedures have the advantage of lower computation cost, parallelizability and no need for feedback. Depending on whether or not the noisy channel used to corrupt the noiseless responses depends on the queries, the noisy 20 questions problem is classified into two categories: querying with measurement-independent noise (e.g.,~\cite{jedynak2012twenty,chung2018unequal}); and querying with measurement-dependent noise (e.g.,~\cite{kaspi2018searching,lalitha2018improved}). As argued in \cite{kaspi2018searching}, measurement-dependent noise can be a better model in many practical applications. For example, for target localization with a sensor network, the noisy response to each query can depend on the size of the query region. Another example is in human query systems where personal biases about the target state may affect the response.

In earlier works on the noisy 20 questions problem, e.g.,~\cite{jedynak2012twenty,tsiligkaridis2014collaborative,tsiligkaridis2015decentralized}, the queries were designed to minimize the entropy of the posterior distribution of the target variable $S$. As pointed out in later works, e.g., \cite{chung2018unequal,chiu2016sequential,kaspi2018searching,lalitha2018improved,chung2017bounds}, other accuracy measures, such as the resolution and the quadratic loss are often better criteria, where the resolution is defined as the absolute difference between $S$ and its estimate $\hatS$, $|\hatS-S|$, and the quadratic loss is $(\hatS-S)^2$. In particular, in estimation problems, if one aims to minimize the differential entropy of the posterior uncertainty of the target variable, then any two queries which can reduce the entropy by the same amount are deemed equally important, even if one query achieves higher estimation accuracy. For example, one query might ask about the most significant bit of the binary expansion of the target variable while the other query might ask about a much less significant bit. These two queries induce equal reductions in the entropy of the posterior distribution for a uniformly distributed target variable. By using the resolution or the quadratic loss, which are directly related with the estimation error, to drive the query design, such ambiguity is avoided. Relations between resolution and entropy were quantified by the bounds in \cite[Theorem 1]{chung2017bounds}.

\subsection{Our Contributions}
Motivated by the scenario of limited resources, computation and response time, we obtain new results on the non-asymptotic tradeoff among the number of queries $n$, the achievable resolution $\delta$ and the excess-resolution probability $\varepsilon$ of optimal adaptive and non-adaptive query procedures for the following noisy 20 questions problem: 
estimation of the location of a target random vector $\bS=(S_1,\ldots,S_d)$ with arbitrary distribution on the unit cube of dimension $d$, i.e., $[0,1]^d$. For the case of adaptive query procedures, we derive an achievable second-order asymptotic bound on the optimal resolution, showing achievability by using an adaptive query procedure based on the variable length feedback code in \cite[Definition 1]{polyanskiy2011feedback}. We define the benefit of adaptivity, called adaptivity gain, as the logarithm of the ratio between achievable resolutions of optimal non-adaptive and adaptive query procedures. This benefit of adaptivity can be attributed to the more informative number of bits extracted by optimal adaptive querying in the binary expansion of each dimension of the target variable. We numerically evaluate a lower bound on the adaptivity gain for measurement-dependent versions of binary symmetric, binary erasure and binary Z- channels.

The main focus of this paper is on non-adaptive query procedures. Our contributions for the case of non-adaptive querying are as follows. Firstly, we derive non-asymptotic resolution bounds of optimal non-adaptive query procedures for arbitrary number of queries $n$ and any excess-resolution probability $\varepsilon$. To do so, similarly to \cite{kaspi2018searching}, we exploit the connection between the 20 questions problem and the channel coding problem. This allows us to borrow ideas from finite blocklength analyses for channel coding~\cite{polyanskiy2010finite} (see also \cite{TanBook}). In particular, we adopt the change-of-measure technique of \cite{csiszar2011information} in the achievability proof to handle the case of measurement-dependent noise.

Secondly, applying the Berry-Esseen theorem, under mild conditions on the measurement-dependent noise, we obtain a second-order asymptotic approximation to the achievable resolution of optimal non-adaptive query procedures with finite number of queries. A key implication of our result states that searching separately for each dimension is actually suboptimal in the regime of a finitely many queries while such a query scheme is optimal in the regime of an infinite number of queries (see the final remark for Theorem \ref{result:second} and the numerical example in Figure \ref{sim_non_adap_sep}). As a corollary, we establish a phase transition for optimal non-adaptive query procedures. This implies that, if one is allowed to make an infinite number of optimal non-adaptive queries, regardless of the excess-resolution probability, the asymptotic average number of bits (in the binary expansion of each dimension of the target variable) extracted per query remains the same. 

We specialize our second-order analyses to three measurement-dependent channel models: the binary symmetric, the binary erasure and the binary Z- channels. Similarly to our proofs for measurement-dependent channels, the second-order asymptotic approximation to the achievable resolution of optimal non-adaptive query procedures for measurement-independent channels is obtained. 

\subsection{Comparison to Previous Work}
\label{sec:comp}

Here we compare the contributions of our paper to related work in the literature~\cite{kaspi2018searching,chiu2016sequential}. First of all, our results hold for arbitrary discrete channels, while the results in \cite{kaspi2018searching,chiu2016sequential} were only established for a measurement-dependent binary symmetric channel. Furthermore, we consider a multidimensional target while the the results in \cite{kaspi2018searching,chiu2016sequential} were only established for a one-dimensional target. In the following, we compare our results, specialized to a single one-dimensional target, with the results in \cite{kaspi2018searching,chiu2016sequential}.

In terms of our results on resolution of non-adaptive query schemes, the most closely related work is \cite{kaspi2018searching}. The authors in \cite{kaspi2018searching} derived first-order asymptotic characterizations of the resolution decay rate when the excess-resolution probability vanishes for a measurement-dependent binary symmetric channel. Our results in Theorem \ref{result:second} extends \cite[Theorem 1]{kaspi2018searching} in several directions. First, Theorem \ref{result:second} is a  second-order asymptotic result which provides an approximation to the performance of optimal query procedures employing a finite number of queries, while \cite[Theorem 1]{kaspi2018searching} gives a first-order asymptotic result when the number of queries tends to infinity. Second, our results hold for any measurement-dependent channel satisfying a mild condition while \cite[Theorem 1]{kaspi2018searching} only considers the measurement-dependent binary symmetric channel (cf. Definition \ref{def:mdBSC}). Furthermore, our results apply methods recently developed for finite blocklength information theory. This results in the first non-asymptotic bounds (cf. Theorems \ref{ach:fbl} and \ref{fbl:converse}) for non-adaptive query schemes for 20 questions search. These bounds extend the analysis of \cite{kaspi2018searching}, in which the derived lower bound on the decay rate of the excess-resolution probability is only tight in the asymptotic limit of large $n$ (e.g., infinite number of queries). Other works concerning non-adaptive query procedures~\cite{jedynak2012twenty,pelc2002searching,chung2018unequal} consider either different performance criteria or different models and thus not comparable to our work.

For resolution limits of adaptive querying, the most closely related publications are \cite{chiu2016sequential,kaspi2018searching,lalitha2018improved}. The authors in \cite{lalitha2018improved} considered the case where the noise is measurement-dependent Gaussian noise. However, the setting in \cite{lalitha2018improved} is different from ours. In \cite{kaspi2018searching}, the authors considered two different adaptive query procedures using ideas due to Forney \cite{forney1968exponential} and Yamamoto-Itoh~\cite{yamamoto1979asymptotic} respectively. The authors of \cite{kaspi2018searching} derived a lower bound on the exponent of the excess-resolution probability for both procedures and showed that the performance of the three-stage adaptive query procedure based on Yamamoto-Itoh~\cite{yamamoto1979asymptotic} has better performance. Furthermore, in \cite{kaspi2018searching}, an asymptotic upper bound on the average number of queries is derived given a particular target resolution and excess-resolution probability~\cite[Theorem 2]{kaspi2018searching}. In \cite{chiu2016sequential}, the authors proposed an adaptive query procedure using sorted posterior matching and derived a non-asymptotic upper bound on the average number of queries subject to a constraint on the excess-resolution probability with respect to a given resolution. In contrast, we present an adaptive query procedure using the ideas in \cite{polyanskiy2011feedback} on finite blocklength analysis for channel coding with feedback and we derive a non-asymptotic upper bound on the excess-resolution probability with respect to a given resolution subject to a constraint on the average number of queries (cf. Theorem \ref{fbl:ach:adaptive}).

It would be interesting to compare our non-asymptotic achievability bound in Theorem \ref{fbl:ach:adaptive} to the results in \cite[Theorem 2]{kaspi2018searching}, \cite{chiu2016sequential}. However, since the respective results are derived under different theoretical assumptions, an analytical comparison is challenging. Note that Theorem \ref{fbl:ach:adaptive} addresses the decay rate of the achievable resolution subject to constraints on the average number of queries and an excess-resolution probability. In contrast, the results in \cite{kaspi2018searching,chiu2016sequential} address an upper bound on the average number of queries subject to a given resolution and an excess-resolution probability constraint. It is difficult to transform our results to an upper bound on the average number of queries or to transform their results to a lower bound on the decay rate of achievable resolution.  comparison of our result with the performance of state-of-the-art algorithms will be presented in a future paper.

\section{Problem Formulation}
\subsection*{Notation}
Random variables and their realizations are denoted by upper case variables (e.g.,  $X$) and lower case variables (e.g.,  $x$), respectively. All sets are denoted in calligraphic font (e.g.,  $\mathcal{X}$). Let $X^n:=(X_1,\ldots,X_n)$ be a random vector of length $n$. We use $\Phi^{-1}(\cdot)$ to denote the inverse of the cumulative distribution function (cdf) of the standard Gaussian. We use $\bbR$, $\bbR_+$ and $\bbN$ to denote the sets of real numbers, positive real numbers and integers respectively. Given any two integers $(m,n)\in\bbN^2$, we use $[m:n]$ to denote the set of integers $\{m,m+1,\ldots,n\}$ and use $[m]$ to denote $[1:m]$. Given any $(m,n)\in\bbN^2$, for any $m$ by $n$ matrix $\ba=\{a_{i,j}\}_{i\in[m],j\in[n]}$, the infinity norm is defined as $\|\ba\|_{\infty}:=\max_{i\in[m],j\in[n]}|a_{i,j}|$. The set of all probability distributions on a finite set $\calX$ is denoted as $\calP(\calX)$ and the set of all conditional probability distributions from $\calX$ to $\calY$ is denoted as $\calP(\calY|\calX)$. Furthermore, we use $\calF(\calS)$ to denote the set of all probability density functions on a set $\calS$. All logarithms are base $e$. Finally, we use $\bbo()$ to denote the indicator function.

\subsection{Noisy 20 Questions Problem On the Unit Cube}
Consider an arbitrary integer $d\in\bbN$. Let $\bS=(S_1,\ldots,S_d)$ be a continuous random vector defined on the unit cube of dimensional $d$ (i.e., $[0,1]^d$) with arbitrary probability density function (pdf) $f_{\bS}$. Note that any searching problem over a bounded $d$-dimensional region is equivalent to a searching problem over the unit cube of dimension $d$ with normalization in each dimension. 

In the estimation problem formulated under the framework of noisy 20 questions, a player aims to accurately estimate the target random variable $\bS$ by posing a sequence of queries $\calA^n=(\calA_1,\ldots,\calA_n)\subseteq[0,1]^{nd}$ to an oracle knowing $\bS$. After receiving the queries, the oracle finds binary answers $\{X_i=\bbo(\bS\in\calA_i)\}_{i\in[n]}$ and passes these answers through a measurement-dependent channel with transition matrix $P_{Y^n|X^n}^{\calA^n}\in\calP(\calY^n|\{0,1\}^n)$ yielding noisy responses $Y^n=(Y_1,\ldots,Y_n)$. Given the noisy responses $Y^n$, the player uses a decoding function $g:\calY^n\to[0,1]^d$ to obtain an estimate $\hat{\bS}=(\hatS_1,\ldots,\hatS_d)$ of the target variable $\bS=(S_1,\ldots,S_d)$. Throughout the paper, we assume that the alphabet $\calY$ for the noisy response is finite.

A query procedure for the noisy 20 questions problem consists of the Lebesgue measurable query sets $\calA^n\subseteq[0,1]^{nd}$ and a decoder $g:\calY^n\to[0,1]^d$. In general, these procedures can be classified into two categories: non-adaptive and adaptive querying. In a non-adaptive query procedure, the player needs to first determine the number of queries $n$ and then design all the queries $\calA^n$ simultaneously. In contrast, in an adaptive query procedure, the design of queries is done sequentially and the number of queries is a variable. In particular, when designing the $i$-th query, the player can use the previous queries and the noisy responses from the oracle to these queries, i.e., $\{\calA_j,Y_j\}_{j\in[i-1]}$, to formulate the next query $\calA_i$. Furthermore, the player needs to choose a stopping criterion, which may be random, determining the number of queries to make. 

We illustrate the difference between non-adaptive and adaptive query procedures in Figure \ref{illustrate:procedures}. In subsequent sections, we clarify the notion of the measurement-dependent channel with concrete examples and present specific definitions of non-adaptive and adaptive query procedures.
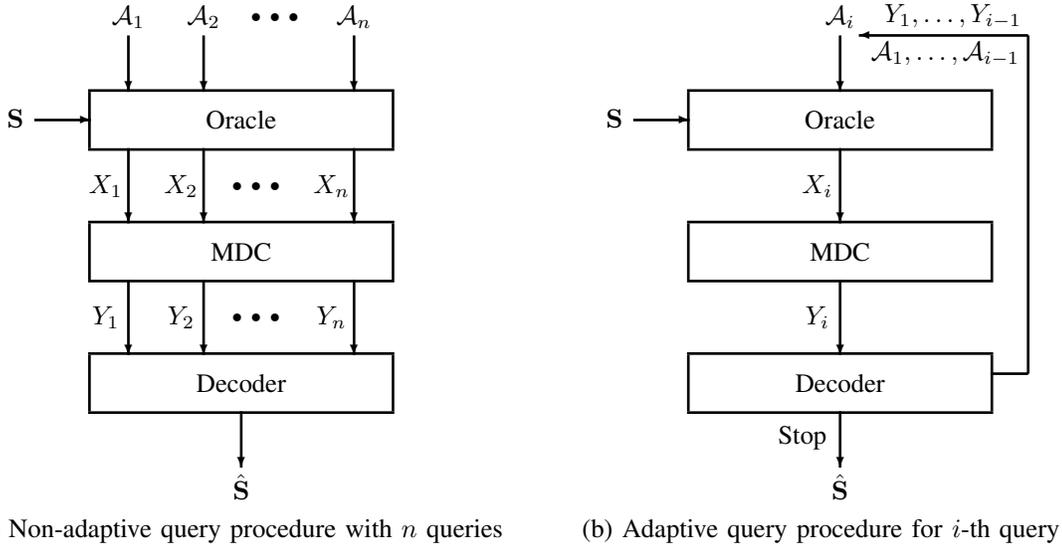
\begin{figure}[tb]
\centering
\setlength{\unitlength}{0.5cm}
\begin{tabular}{cc}
\scalebox{1}{
\begin{picture}(8,12.5)
\linethickness{1pt}
\put(1,12.5){\makebox(0,0){$\calA_1$}}
\put(3,12.5){\makebox(0,0){$\calA_2$}}
\put(4.5,12.5){\makebox(0,0){\circle*{0.2}}}
\put(5,12.5){\makebox(0,0){\circle*{0.2}}}
\put(5.5,12.5){\makebox(0,0){\circle*{0.2}}}
\put(7,12.5){\makebox(0,0){$\calA_n$}}
\put(1,12){\vector(0,-1){1.5}}
\put(3,12){\vector(0,-1){1.5}}
\put(7,12){\vector(0,-1){1.5}}
\put(-2,9.8){\makebox(0,0){$\bS$}}
\put(-1.5,9.75){\vector(1,0){1.5}}
\put(0,9){\framebox(8,1.5){Oracle}}
\put(1,9){\vector(0,-1){2}}
\put(3,9){\vector(0,-1){2}}
\put(7,9){\vector(0,-1){2}}
\put(0.4,8){\makebox(0,0){$X_1$}}
\put(2.4,8){\makebox(0,0){$X_2$}}
\put(4,8){\makebox(0,0){\circle*{0.2}}}
\put(4.5,8){\makebox(0,0){\circle*{0.2}}}
\put(5,8){\makebox(0,0){\circle*{0.2}}}
\put(6.4,8){\makebox(0,0){$X_n$}}
\put(0,5.5){\framebox(8,1.5){MDC}}
\put(1,5.5){\vector(0,-1){2}}
\put(3,5.5){\vector(0,-1){2}}
\put(7,5.5){\vector(0,-1){2}}
\put(0.4,4.5){\makebox(0,0){$Y_1$}}
\put(2.4,4.5){\makebox(0,0){$Y_2$}}
\put(4,4.5){\makebox(0,0){\circle*{0.2}}}
\put(4.5,4.5){\makebox(0,0){\circle*{0.2}}}
\put(5,4.5){\makebox(0,0){\circle*{0.2}}}
\put(6.4,4.5){\makebox(0,0){$Y_n$}}
\put(0,2){\framebox(8,1.5){Decoder}}
\put(4,2){\vector(0,-1){1.5}}
\put(4,0){\makebox(0,0){$\hat{\bS}$}}
\end{picture}}
&
\hspace{.4in}
\scalebox{1}{
\begin{picture}(8,12.5)
\linethickness{1pt}
\put(4,12.5){\makebox(0,0){$\calA_i$}}
\put(4,12){\vector(0,-1){1.5}}
\put(-2,9.8){\makebox(0,0){$\bS$}}
\put(-1.5,9.75){\vector(1,0){1.5}}
\put(0,9){\framebox(8,1.5){Oracle}}
\put(4,9){\vector(0,-1){2}}
\put(3.4,8){\makebox(0,0){$X_i$}}
\put(0,5.5){\framebox(8,1.5){MDC}}
\put(4,5.5){\vector(0,-1){2}}
\put(3.4,4.5){\makebox(0,0){$Y_i$}}
\put(9,3){\line(0,1){9}}
\put(9,12){\vector(-1,0){4.5}}
\put(7,12.5){\makebox(0,0){$Y_1,\ldots,Y_{i-1}$}}
\put(6.8,11.5){\makebox(0,0){$\calA_1,\ldots,\calA_{i-1}$}}
\put(0,2){\framebox(8,1.5){Decoder}}
\put(8,3){\line(1,0){1}}
\put(4,2){\vector(0,-1){1.5}}
\put(3,1.3){\makebox(0,0){Stop}}
\put(4,0){\makebox(0,0){$\hat{\bS}$}}
\end{picture}}
\vspace{.1in}
\\
(a) Non-adaptive query procedure with $n$ queries&\hspace{.2in} (b) Adaptive query procedure for $i$-th query
\end{tabular}
\caption{Illustration of query procedures for the noisy 20 questions problem with measurement-dependent channel (MDC). In the non-adaptive case (a), a target slate $\bS$ is known to the oracle who responds to a block of queries $\calA_1,\ldots,\calA_n$ and provides binary responses $X_1,\ldots,X_n$, respectively. These responses are corrupted by a measurement-dependent channel (MDC) that outputs symbols $Y_1,\ldots,Y_n$, which are used by the decoder to produce estimate $\hat{\bS}$. In the adaptive case (b), the queries are posed sequentially and decoder needs to determine when to stop the query procedure.
}
\label{illustrate:procedures}
\end{figure}

\subsection{The Measurement-Dependent Channel}
In this subsection, we describe succinctly the measurement-dependent channel scenario~\cite{kaspi2018searching}, also known as a channel with state~\cite[Chapter 7]{el2011network}. Given a sequence of queries $\calA^n\subseteq[0,1]^{nd}$, the channel from the oracle to the player is a memoryless channel whose transition probabilities are functions of the queries. Specifically, for any $(x^n,y^n)\in\{0,1\}^n\times\calY^n$,
\begin{align}
P_{Y^n|X^n}^{\calA^n}(y^n|x^n)
&=\prod_{i\in[n]}P_{Y|X}^{\calA_i}(y_i|x_i),
\end{align}
where $P_{Y|X}^{\calA_i}$ denotes the transition probability of the channel which depends on the $i$-th query $\calA_i$. Given any query $\calA\subseteq[0,1]^d$, define the volume $|\calA|$ of $\calA$ as its Lebesgue measure, i.e., $|\calA|=\int_{t\in\calA}\rmd t$. Throughout the paper, we consider only Lebesgue measurable query sets and 
assume that the measurement-dependent channel $P_{Y|X}^{\calA}$ depends on the query $\calA$ only through its size. Thus, $P_{Y|X}^{\calA}$ is equivalent to a channel with state $P_{Y|X}^q$ where the state $q=|\calA|\in[0,1]$.

For any $q\in[0,1]$, any $\xi\in(0,\min(q,1-q))$ and any subsets $\calA$, $\calA^+$ and $\calA^-$ of $[0,1]$ with sizes $|\calA|=q$, $|\calA^+|=q+\xi$ and $|\calA^-|=q-\xi$, we assume the measurement-dependent channel is continuous in the sense that there exists a constant $c(q)$ depending on $q$ only such that
\begin{align}
\max\left\{\left\|\log\frac{P_{Y|X}^{\calA}}{P_{Y|X}^{{\calA^+}}}\right\|_{\infty},\left\|\log\frac{P_{Y|X}^\calA}{P_{Y|X}^{\calA^-}}\right\|_{\infty}\right\}\leq c(q)\xi\label{assump:continuouschannel},
\end{align}
where the infinity norm is defined as $\|\ba\|_{\infty}:=\max_{i\in[m],j\in[n]}|a_{i,j}|$ for any matrix $\ba=\{a_{i,j}\}_{i\in[m],j\in[n]}$.

Some examples of measurement-dependent channels satisfying the continuous constraint in \eqref{assump:continuouschannel} are as follows.
\begin{definition}
\label{def:mdBSC}
Given any $\calA\subseteq[0,1]$, a channel $P_{Y|X}^{\calA}$ is said to be a measurement-dependent Binary Symmetric Channel (BSC) with parameter $\nu\in[0,1]$ if $\calX=\calY=\{0,1\}$ and 
\begin{align}
P_{Y|X}^{\calA}(y|x)=(\nu|\calA|)^{\bbo(y\neq x)}(1-\nu|\calA|)^{\bbo(y=x)},~\forall~(x,y)\in\{0,1\}^2.
\end{align}
\end{definition}
This definition generalizes \cite[Theorem 1]{kaspi2018searching}, where the authors considered a measurement-dependent BSC with parameter $\nu=1$. Note that the binary output bit of a measurement-dependent BSC with parameter $\nu$ is flipped with probability $\nu|\calA|$.

\begin{definition}
Given any $\calA\subseteq[0,1]$, a measurement-dependent channel $P_{Y|X}^{\calA}$ is said to be a measurement-dependent Binary Erasure Channel (BEC) with parameter $\tau\in[0,1]$ if $\calX=\{0,1\}$, $\calY=\{0,1,\rme\}$ and
\begin{align}
P_{Y|X}^{\calA}(y|x)=(1-\tau|\calA|)^{\bbo(y=x)}(\tau|\calA|)^{\bbo(y=\rme)}
\end{align}
\end{definition}
Note that the binary output bit of a measurement-dependent BEC with parameter $\tau$ is erased with probability $\tau|\calA|$.

\begin{definition}
\label{def:mdZ}
Given any $\calA\subseteq[0,1]$, a measurement-dependent channel $P_{Y|X}^{\calA}$ is said to be a measurement-dependent Z-channel with parameter $\zeta\in[0,1]$ if $\calX=\{0,1\}$, $\calY=\{0,1\}$ and
\begin{align}
P_{Y|X}^{\calA}(y|x)=(1-\zeta|\calA|)^{\bbo(y=x=1)}(\zeta|\calA|)^{\bbo(y=0,x=1)}(0)^{\bbo(y=1,x=0)}.
\end{align}
\end{definition}
Note that the binary output bit of a measurement-dependent Z-channel is flipped with probability $\zeta|\calA|$ if the input is $x=1$.

Each of these measurement-dependent channels will be considered in the sequel.

In contrast, in a measurement-independent channel, the noisy channel that corrupts the noiseless response remains the same regardless of the query, i.e., $P_{Y|X}^\calA=P_{Y|X}$ for any $\calA\subseteq[0,1]$, where $P_{Y|X}\in\calP(\calY|\calX)$ is a given channel with the input alphabet $\calX=\{0,1\}$ and the finite output alphabet $\calY$.

\subsection{Non-Adaptive Query Procedures}
A non-adaptive query procedure with resolution $\delta$ and excess-resolution constraint $\varepsilon$ is defined as follows.

\begin{definition}
\label{def:procedure}
Given any $(n,d)\in\bbN^2$, $\delta\in\bbR_+$ and $\varepsilon\in[0,1]$, an $(n,d,\delta,\varepsilon)$-non-adaptive query procedure for the noisy 20 questions consists of 
\begin{itemize}
\item $n$ queries $(\calA_1,\ldots,\calA_n)$ where each $\calA_i\subseteq[0,1]^d$,
\item and a decoder $g:\calY^n\to[0,1]^d$
\end{itemize}
such that the excess-resolution probability satisfies
\begin{align}
\rmP_\rme(n,d,\delta)&:=\sup_{f_{\bS}\in\calF([0,1]^d)}\Pr\{\exists~i\in[d]:~|\hatS_i-S_i|>\delta\}\leq \varepsilon\label{def:excessresolution2},
\end{align}
where $\hatS_i$ is the estimate of $i$-th element of the $d$-dimensional target $\bS$ using the decoder $g$, i.e., $g(Y^n)=(\hatS_1,\ldots,\hatS_d)$.
\end{definition} 

In Algorithm \ref{procedure:nonadapt}, we provide a non-adaptive query procedure which is used in our achievability proof. The procedure is parametrized by two parameters $M$ and $p$, where $\frac{1}{M}$ is the target resolution and $p$ is the design parameter. The definition of the \emph{excess-resolution probability} with respect to $\delta$ is inspired by rate-distortion theory~\cite{berger1971rate,kostina2013lossy}. Our formulation generalizes that of \cite{kaspi2018searching} where the authors constrained the target-dependent maximum excess-resolution probability for the case of $d=1$, i.e., $\sup_{s_1\in[0,1]}\Pr\{|\hatS_1-s_1|>\delta\}$.

In practical applications, the number of queries are often limited to minimize total cost of queries and maintain low latency. We are interested in the establishing a non-asymptotic fundamental limit to achievable resolution $\delta$:
\begin{align}
\delta^*(n,d,\varepsilon)
&:=\inf\big\{\delta\in[0,1]:\exists\mathrm{~an~}(n,d,\delta,\varepsilon)\mathrm{-non}\mathrm{-adaptive}\mathrm{~query}\mathrm{~procedure}\big\}\label{def:delta*}.
\end{align}
Note that $\delta^*(n,d,\varepsilon)$ denotes the minimal resolution one can achieve with probability at least $1-\varepsilon$ using a non-adaptive query procedure with $n$ queries. In other words, $\delta^*(n,d,\varepsilon)$ is the achievable resolution of optimal non-adaptive query procedures tolerating an excess-resolution probability of $\varepsilon\in[0,1]$. Dual to \eqref{def:delta*} is the sample complexity, determined by the minimal number of queries required to achieve a resolution $\delta$ with probability at least $1-\varepsilon$, i.e.,
\begin{align}
n^*(d,\delta,\varepsilon):=\inf\big\{n\in\bbN:\exists\mathrm{~an~}(n,d,\delta,\varepsilon)\mathrm{-non}\mathrm{-adaptive}\mathrm{-query}\mathrm{-procedure}\big\}\label{def:n*}.
\end{align}
One can easily verify that for any $(\delta,\varepsilon)\in\bbR_+\times[0,1]$,
\begin{align}
n^*(d,\delta,\varepsilon)
&=\inf\{n:\delta^*(n,d,\varepsilon)\leq \delta\}\label{def:sc_non}.
\end{align}
Thus, it suffices to derive the fundamental limit $\delta^*(n,d,\varepsilon)$.

Note that in Definition \ref{def:procedure}, the probability $\Pr\{\exists~i\in[d]:~|\hatS_i-S_i|>\delta\}$ is equivalent to $\Pr\{\max_{i\in[d]}|\hatS_i-S_i|>\delta\}$. As contrasted to the $L_\infty$ norm, i.e., $\|\hat{\bS}-\bS\|_\infty=\max_{i\in[d]}|\hatS_i-S_i|$, we could have considered the $L_2$ norm, i.e., $\|\hat{\bS}-\bS\|_2=\sqrt{\sum_{i\in[d]}(\hatS_i-S_i)^2}$, for which a corresponding fundamental limit $\delta^*_{L_2}(n,d,\varepsilon)$ could be obtained. In fact, we have $\delta^*(n,d,\varepsilon)\leq \delta^*_{L_2}(n,d,\varepsilon)\leq \sqrt{d}\delta^*(n,d,\varepsilon)$. This is because for any location vector $\bS=(S_1,\ldots,S_d)\in[0,1]^d$ and any estimated vector $\hat{\bS}=(\hatS_1,\ldots,\hatS_d)\in[0,1]^d$, $\frac{\|\hat{\bS}-\bS\|_2}{\sqrt{d}}\leq\|\hat{\bS}-\bS\|_\infty\leq \|\hat{\bS}-\bS\|_2$ and therefore
\begin{align}
\Pr\left\{\|\hat{\bS}-\bS\|_2>\sqrt{d}\delta\right\}\leq \Pr\{\|\hat{\bS}-\bS\|_\infty>\delta\}\leq \Pr\{\|\hat{\bS}-\bS\|_2>\delta\}.
\end{align}
On the one hand, if a query procedure is $(n,d,\delta,\varepsilon)$-achievable under the $L_\infty$ norm criterion, then since $\Pr\{\|\hat{\bS}-\bS\|_2>\sqrt{d}\delta\}\leq \Pr\{\|\hat{\bS}-\bS\|_\infty>\delta\}$, the query procedure is $(n,d,\sqrt{d}\delta,\varepsilon)$-achievable under the $L_2$ norm criterion and thus $\delta^*_{L_2}(n,d,\delta,\varepsilon)\leq \sqrt{d}\delta^*(n,d,\varepsilon)$. On the other hand, if a query procedure is $(n,d,\delta,\varepsilon)$-achievable under the $L_2$ norm criterion, then since $\Pr\{\|\hat{\bS}-\bS\|_\infty>\delta\}\leq \Pr\{\|\hat{\bS}-\bS\|_2>\delta\}$, the query procedure is $(n,d,\delta,\varepsilon)$-achievable under the $L_\infty$ norm criterion and thus $\delta^*(n,d,\varepsilon)\leq\delta^*_{L_2}(n,d,\varepsilon)$.

\subsection{Adaptive Query Procedures}

An adaptive query procedure with resolution $\delta$ and excess-resolution constraint $\varepsilon$ is defined as follows.
\begin{definition}
\label{def:adaptive:procedure}
Given any $(l,d,\delta,\varepsilon)\in\bbR_+\times\bbN\times\bbR_+\times[0,1]$, an $(l,d,\delta,\varepsilon)$-adaptive query procedure for the noisy 20 questions problem consists of
\begin{itemize}
\item a sequence of adaptive queries where for each $i\in\bbN$, the design of query $\calA_i\subseteq[0,1]^d$ is based all previous queries $\{\calA_j\}_{j\in[i-1]}$ and the noisy responses $Y^{i-1}$ from the oracle
\item a sequence of decoding functions $g_i:\calY^i\to[0,1]^d$ for $i\in\bbN$ 
\item a random stopping time $\tau$ depending on noisy responses $\{Y_i\}_{i\in\bbN}$ such that under any pdf $f_{\bS}$ of the target random variable $\bS$, the average number of queries satisfies
\begin{align}
\mathbb{E}[\tau]\leq l,
\end{align}
\end{itemize}
such that the excess-resolution probability satisfies
\begin{align}
\rmP_{\rme,\rma}(l,d,\delta):=\sup_{f_{\bS}\in\calF([0,1]^d)}\Pr\{\exists~i\in[d]:~|\hatS_i-S_i|>\delta\}\leq \varepsilon\label{def:excessresolution},
\end{align}
where $\hatS_i$ is the estimate of $i$-th element of the target $\bS$ using the decoder $g$ at time $\tau$, i.e., $g(Y^\tau)=(\hatS_1,\ldots,\hatS_d)$.
\end{definition}

Similar to \eqref{def:delta*}, given any $(l,d,\varepsilon)\in\bbR_+\times\bbN\times[0,1)$, we can define the fundamental resolution limit for adaptive querying as follows: 
\begin{align}
\delta_\rma^*(l,d,\varepsilon)
&:=\inf\{\delta\in\bbR_+:~\exists~\mathrm{an}~(l,d,\delta,\varepsilon)\mathrm{-adaptive}~\mathrm{query~procedure}\}\label{def:delta*:adaptive},
\end{align}
with analogous definition of mean sample complexity (cf. \eqref{def:sc_non})
\begin{align}
l^*(d,\delta,\varepsilon):=\inf\{l\in\bbR_+:~\exists~\mathrm{an}~(l,d,\delta,\varepsilon)\mathrm{-adaptive}~\mathrm{query~procedure}\}.
\end{align}

\section{Main Results for Non-Adaptive Query Procedures}

\subsection{Non-Asymptotic Bounds}
We first present an upper bound on the error probability of optimal non-adaptive query procedures. Given any $(p,q)\in[0,1]^2$, let $P_Y^{p,q}$ be the marginal distribution on $\calY$ induced by the Bernoulli distribution $P_X=\mathrm{Bern}(p)$ and the measurement-dependent channel $P_{Y|X}^{q}$. Furthermore, define the following information density
\begin{align}
\imath_{p,q}(x;y)&:=\log\frac{P_{Y|X}^q(y|x)}{P_Y^{p,q}(y)},~\forall~(x,y)\in\calX\times\calY\label{def:ipqxy}.
\end{align}
Correspondingly, for any $(x^n,y^n)\in\calX^n\times\calY^n$, we define
\begin{align}
\imath_p(x^n;y^n)
&:=\sum_{i\in[n]}\imath_{p,p}(x_i;y_i)\label{def:ixnyn}
\end{align}
as the mutual information density between $x^n$ and $y^n$.

\begin{algorithm}[bt]
\caption{Non-adaptive query procedure for searching for a multidimensional target over the unit cube}
\label{procedure:nonadapt}
\begin{algorithmic}
\REQUIRE The number of queries $n\in\bbN$, the dimension $d\in\bbN$ and two parameters $(M,p)\in\bbN\times(0,1)$
\ENSURE An estimate $(\hats_1,\ldots,\hats_d)\in[0,1]^d$ of a $d$-dimensional target variable $(s_1,\ldots,s_d)\in[0,1]^d$
\STATE Partition the unit cube of dimension $d$ (i.e., $[0,1]^d$) into $M^d$ equal-sized disjoint cubes $\{\calS_{i_1,\ldots,i_d}\}_{(i_1,\ldots,i_d)\in[M]^d}$.
\STATE Generate $M^d$ binary vectors $\{x^n(i_1,\ldots,i_d)\}_{(i_1,\ldots,i_d)\in[M]^d}$ where each binary vector is generated i.i.d. from a Bernoulli distribution with parameter $p$.
\STATE $t \leftarrow 1$.
\WHILE{$t\leq n$}
\STATE Form the $t$-th query as
\begin{align*}
\calA_t:=\bigcup_{(i_1,\ldots,i_d)\in[M]^d:x_t(i_1,\ldots,i_d)=1}\calS_{i_1,\ldots,i_d}.
\end{align*}
\STATE Obtain a noisy response $y_t$ from the oracle to the query $\calA_t$.
\STATE $t \leftarrow t+1$.
\ENDWHILE
\STATE Generate estimates $(\hats_1,\ldots,\hats_d)$ as
\begin{align*}
\hats_i=\frac{2\hatw_i-1}{2M},~i\in[d]
\end{align*}
where $\hat{\bw}=(\hatw_1,\ldots,\hatw_d)$ is obtained via the maximum mutual information density estimator, i.e,
\begin{align*}
\hat{\bw}=\argmax_{(\tili_1,\ldots,\tili_d)\in[M]^d}\imath_p(x^n(\tili_1,\ldots,\tili_d);y^n).
\end{align*}
\end{algorithmic}
\end{algorithm}

\begin{theorem}
\label{ach:fbl}
Given any $(n,d,M)\in\bbN^3$, for any $p\in[0,1]$ and any $\eta\in\bbR_+$, the procedure in Algorithm \ref{procedure:nonadapt} is an $(n,d,\frac{1}{M},\varepsilon)$-non-adaptive query procedure where
\begin{align}
\varepsilon&\leq 
4n\exp(-2M^d\eta^2)+\exp(n\eta c(p))\mathbb{E}[\min\{1,M^d\Pr\{\imath_p(\barX^n;Y^n)\geq \imath_p(X^n;Y^n)|X^n,Y^n\}]\}\label{com:fbl},
\end{align}
where $(X^n,\barX^n,Y^n)$ is distributed as $P_X^n(X^n)P_X^n(\barX^n)(P_{Y|X}^p)^n(Y^n|X^n)$ with $P_X$ defined as the Bernoulli distribution with parameter $p$ (i.e., $P_X(1)=p$).
\end{theorem}
The proof of Theorem \ref{ach:fbl} uses a modification of the random coding union bound~\cite{polyanskiy2010finite} and is given in Appendix \ref{proof:ach}. 

Consider the measurement-independent channel where $P_{Y|X}^q=P_{Y|X}^{1}=:P_{Y|X}$ for all $q\in[0,1]$. Similarly to the proof of Theorem \ref{ach:fbl}, we can show that for any $p\in[0,1]$, there exists an $(n,d,\frac{1}{M},\varepsilon)$-non-adaptive query procedure such that
\begin{align}
\varepsilon&\leq \mathbb{E}[\min\{1,M^d\Pr\{\jmath_p(\barX^n;Y^n)\geq \jmath_p(X^n;Y^n)|X^n,Y^n\}]\}\label{com:fbl2},
\end{align}
where the tuple of random variables $(X^n,\barX^n,Y^n)$ is distributed as $P_X^n(X^n)P_X^n(\barX^n)P_{Y|X}^n(Y^n|X^n)$, the information density $\jmath_p(x^n;y^n)$ is defined as
\begin{align}
\jmath_p(x^n;y^n):=\log\frac{P_{Y|X}^n(y^n|x^n)}{P_Y^n(y^n)},
\end{align}
and $P_Y$ is the marginal distribution induced by $P_X$ and $P_{Y|X}$. Comparing the measurement-independent case \eqref{com:fbl2} with the measurement-dependent case \eqref{com:fbl}, the non-asymptotic upper bound \eqref{com:fbl} in Theorem \ref{ach:fbl} differs from \eqref{com:fbl2} in two aspects: there are an additional additive term and an additional multiplicative term in \eqref{com:fbl}. As is made clear in the proof of Theorem \ref{ach:fbl}, the additive term $4n\exp(-2M^d\eta^2)$ results from the atypicality of the measurement-dependent channel and the multiplicative term $\exp(n\eta c(p))$ appears due to the change-of-measure we use to replace the measurement-dependent channel $P_{Y^n|X^n}^{\calA^n}$ with the measurement-independent channel $(P_{Y|X}^p)^n$.

We conjecture that the additional exponential terms in the upper bound \eqref{com:fbl} are fundamentally related to the measurement-dependent channel assumption underlying Theorem 1.
However, the truth of this conjecture cannot be established without a matching non-asymptotic lower bound. In our proof of Theorem \ref{ach:fbl}, the upper bound on $\varepsilon$ is established using the random coding based non-adaptive query procedure in Algorithm \ref{procedure:nonadapt}. If one adopts another query procedure where in each query, the query size is the same, it is possible that the additive ``channel atypicality'' term can be removed and other parts of the formulas might also change.
 
We next provide a non-asymptotic converse bound to complement Theorem \ref{ach:fbl}. For simplicity, for any query $\calA\subseteq[0,1]^d$ and any $(x,y)\in\calX\times\calY$, we use $\imath_{\calA}(x,y)$ to denote $\imath_{|\calA|,|\calA|}(x,y)$.

\begin{theorem}
\label{fbl:converse}
Set $(n,\delta,\varepsilon)\in\bbN\times\bbR_+\times[0,1]$. Any $(n,\delta,\varepsilon)$-non-adaptive query procedure satisfies the following. For any $\beta\in(0,\frac{1-\varepsilon}{2})$ and any $\kappa\in(0,1-\varepsilon-2d\beta)$,
\begin{align}
-d\log\delta\leq -d\log\beta-\log\kappa+\sup_{\calA^n\subseteq[0,1]^{nd}}\sup\bigg\{t\Big|\Pr\Big\{\sum_{i\in[n]}\imath_{\calA_i}(X_i;Y_i)\leq t\bigg\}\leq \varepsilon+2d\beta+\kappa\Big\}\label{uppbound}.
\end{align}
\end{theorem}
The proof of Theorem \ref{fbl:converse} is given in Appendix \ref{proof:converse}. The proof of Theorem \ref{fbl:converse} is decomposed into two steps: i) we use the result in \cite{kaspi2018searching} which states that the excess-resolution probability of any non-adaptive query procedure can be lower bounded by the error probability associated with channel coding over the measurement-dependent channel with uniform message distribution, minus a certain term depending on $\beta$; and ii) we apply the non-asymptotic converse bound for channel coding~\cite[Proposition 4.4]{TanBook} by realizing that, given a sequence of queries, the measurement-dependent channel is simply a time varying channel with deterministic states at each time point.

We remark that the non-asymptotic bounds in Theorems \ref{ach:fbl} and \ref{fbl:converse} hold for any number of queries and any measurement-dependent channels satisfying \eqref{assump:continuouschannel}. As we shall see in the next subsection, these non-asymptotic bounds lead to the second-order asymptotic result in Theorem \ref{result:second}, which provides an approximation to the finite blocklength fundamental limit $\delta^*(n,d,\varepsilon)$. The exact calculation of the upper bound in Theorem \ref{fbl:converse} is challenging. However, for $n$ sufficiently large, as demonstrated in the proof of Theorem \ref{result:second}, the supremum in \eqref{uppbound} can be achieved by queries $\calA^n$ where each query $\calA_i$ has the same size.

\subsection{Second-Order Asymptotic Approximation}
In this subsection, we present the second-order asymptotic approximation to the achievable resolution $\delta^*(n,d,\varepsilon)$ of optimal non-adaptive query procedures after $n$ queries subject to a worst case excess-resolution probability of $\varepsilon\in[0,1)$.

Given measurement-dependent channels $\{P_{Y|X}^q\}_{q\in[0,1]}$, the channel ``capacity" is defined as
\begin{align}
C&:=\max_{q\in[0,1]} \mathbb{E}[\imath_{q,q}(X;Y)]\label{def:capacity},
\end{align}
where $(X,Y)\sim \mathrm{Bern}(q)\times P_{Y|X}^q$.

Let the capacity-achieving set $\calP_{\rm{ca}}$ be the set of optimizers achieving \eqref{def:capacity}. Then, for any $\varepsilon\in[0,1)$, define the following ``dispersion'' of the measurement-dependent channel
\begin{align}
V_\varepsilon
&:=\left\{
\begin{array}{cc}
\min_{q\in\calP_{\rm{ca}}}\mathrm{Var}[\imath_{q,q}(X;Y)]&\mathrm{if~}\varepsilon<0.5,\\
\max_{q\in\calP_{\rm{ca}}}\mathrm{Var}[\imath_{q,q}(X;Y)]&\mathrm{if~}\varepsilon\geq 0.5.
\end{array}
\right.
\label{def:veps}
\end{align}
The case of $\varepsilon<0.5$ will be the focus of the sequel of this paper. 

Note that since we consider channels with finite input and output alphabets, i.e., $|\calX|<\infty$ and $|\calY|<\infty$, similarly to \cite[Lemma 46]{polyanskiy2010finite}, the third absolute moment of $\imath_{q,q}(X;Y)$ is finite for any $q\in(0,1)$.
\begin{theorem}
\label{result:second}
For any $\varepsilon\in(0,1)$, the achievable resolution $\delta^*(n,d,\varepsilon)$ of optimal non-adaptive query procedures satisfies
\begin{align}
-\log\delta^*(n,d,\varepsilon)
&=\frac{1}{d}\Big(nC+\sqrt{nV_\varepsilon}\Phi^{-1}(\varepsilon)+O(\log n)\Big)\label{joint:search},
\end{align}
where the remainder term satisfies $-\frac{1}{2}\log n+O(1)\leq O(\log n)\leq\log n+O(1)$.
\end{theorem}
The proof of Theorem \ref{result:second} is provided in Appendix \ref{proof:second}. In the achievability proof, we make use of the non-adaptive query procedure in Algorithm \ref{procedure:nonadapt} and thus prove its second-order asymptotic optimality.

We make the following remarks. Firstly, Theorem \ref{result:second} implies a phase transition analogous to those found in group testing~\cite{Scarlett:2016:PTG:2884435.2884439} and the pooled data problem~\cite{scarlett2017nips}, which we exhibit in Figure \ref{pt}.
\begin{figure}[tb]
\centering
\includegraphics[width=.5\columnwidth]{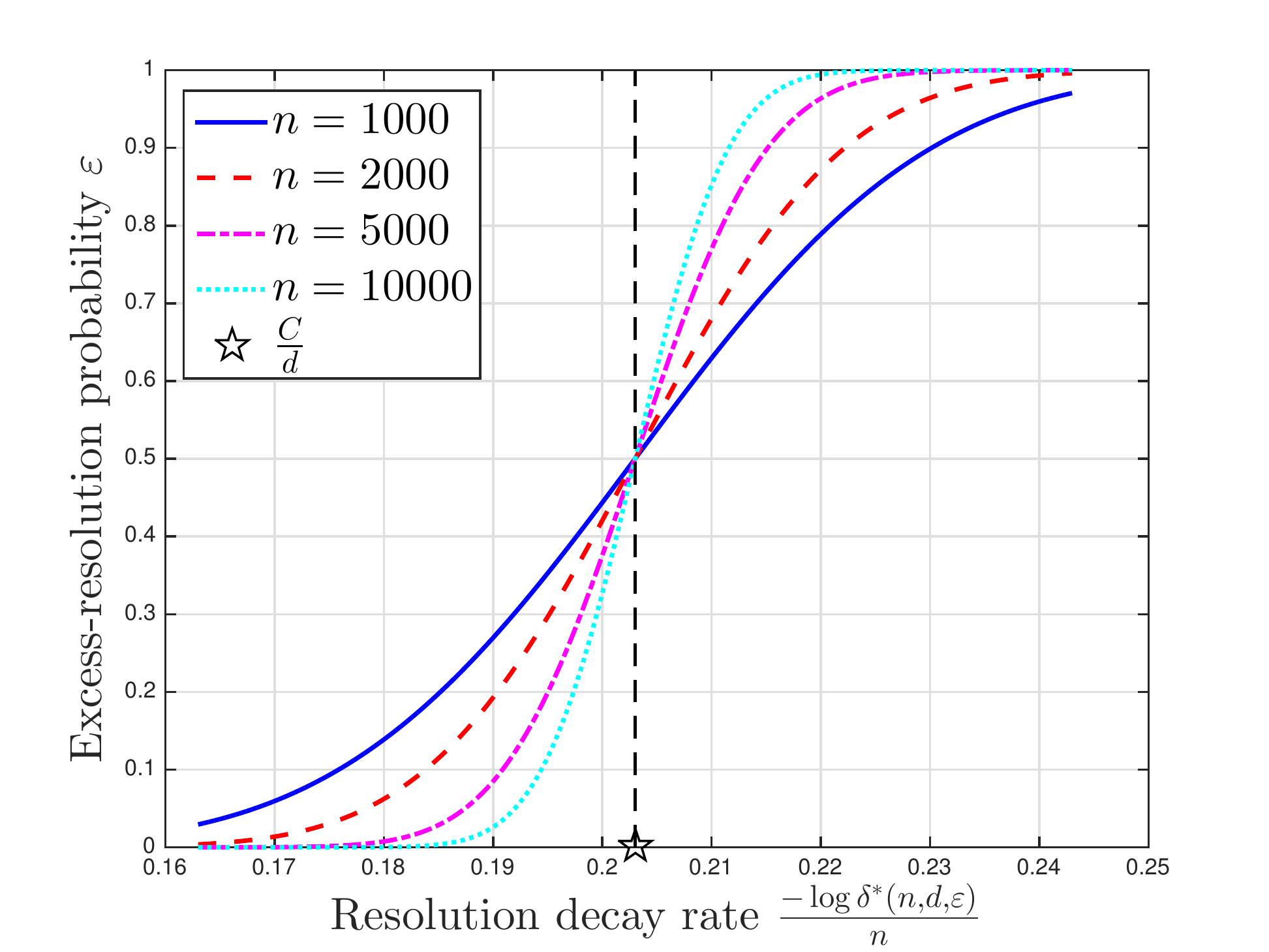}
\caption{Illustration of the phase transition of non-adaptive query procedures for the case of $d=2$ when the noisy channel is a measurement-dependent BSC with parameter $\nu=0.2$. On the one hand, when the resolution decay rate is strictly greater than the capacity $\frac{C}{d}$, then as the number of the queries $n\to\infty$, the excess-resolution probability tends to one. On the other hand, when the resolution decay rate is strictly less than the capacity $\frac{C}{d}$, then the excess-resolution probability vanishes as the number of the queries increases.}
\label{pt}
\end{figure}
We remark that this phase transition is a direct result of the second-order asymptotic analysis and does not follow from a first-order asymptotic analysis, e.g., that was developed in \cite[Theorem 1]{kaspi2018searching}. As a corollary of Theorem \ref{result:second}, for any $\varepsilon\in(0,1)$, we see that the first-order term $\frac{C}{d}$ in the bound \eqref{joint:search} determines the asymptotic convergence rate. Specifically,
\begin{align}
\lim_{n\to\infty}-\frac{1}{n}\log \delta^*(n,d,\varepsilon)=\frac{C}{d}\label{phasetransition}.
\end{align}
which takes the form of a strong converse~\cite{zhou2016cilossy,wei2009strong,liu2016brascamp} to the channel coding theorem. The result in \eqref{phasetransition} indicates that tolerating a smaller, or even vanishing excess-resolution probability, does not improve the asymptotic achievable resolution decay rate of an optimal non-adaptive query procedure.

Secondly, Theorem \ref{result:second} refines \cite[Theorem 1]{kaspi2018searching} in several directions. First, Theorem \ref{result:second} is a second-order asymptotic result that provides an approximation for the finite blocklength performance while \cite[Theorem 1]{kaspi2018searching} only characterizes the asymptotic resolution decay rate with vanishing worst-case excess-resolution probability for a one-dimensional target, i.e., $\lim_{\varepsilon\to 0}\lim_{n\to\infty}(-\log\delta^*(n,1,\varepsilon))$. Second, our results hold for any measurement-dependent channel satisfying \eqref{assump:continuouschannel} while \cite[Theorem 1]{kaspi2018searching} only addresses the measurement-dependent BSC.

Thirdly, the dominant event which leads to an excess-resolution in noisy 20 questions estimation is the atypicality of the information density $\imath_p(X^n;Y^n)$ (cf. \eqref{def:ixnyn}). To characterize the probability of this event, we make use of the Berry--Esseen theorem and show that the mean $C$ and the variance $\rmV_\varepsilon$ of the information density $\imath_{q,q}(X;Y)$ play critical roles.

Fourthly, we remark that any number $s\in[0,1]$ has the binary expansion $(b_0.b_1b_2\ldots)$. We can thus interpret Theorem \ref{result:second} as follows: using an optimal non-adaptive query procedure, after $n$ queries, with probability of at least $1-\varepsilon$, one can extract the first $\lfloor-\log_2\delta^*(n,d,\varepsilon)\rfloor$ bits of the binary expansion of each dimension of the target variable $\bS=(S_1,\ldots,S_d)$.

A final remark is that if the target variable is multi-dimensional (i.e., $d>1$), applying Algorithm \ref{procedure:nonadapt} independently over each dimension will not achieve the second-order asymptotic optimality. In contrast, such an independent application of Algorithm \ref{procedure:nonadapt} to each dimension will achieve the first-order asymptotic optimality in the limit of an infinite number of queries. We explain this dichotomy between first- and second-order asymptotic optimality below.

From the result in \eqref{phasetransition}, we observe that it is in fact first-order asymptotically optimal to allocate roughly $\frac{n}{d}$ queries to each dimension using the special $d=1$ case of Algorithm \ref{procedure:nonadapt} when searching for a $d$-dimensional target variable. Such a decoupled searching algorithm achieves the first-order asymptotic optimal resolution decay rate for non-adaptive query procedures. However, when we include the second-order term in \eqref{joint:search}, it is clear that allocating equal number of queries to search over each dimension is \emph{not} optimal in terms of second-order asymptotics. Supposed that over each dimension $i\in[d]$, we allocate $\frac{n}{d}$ queries to search for the value of $S_i$ and tolerate excess-resolution probabilities $\varepsilon_i$, similarly to the achievability part of Theorem \ref{result:second}, we find that the achievable resolution $\delta_{\rm{sep}}(n,d,\varepsilon)$ satisfies
\begin{align}
-\log\delta_{\rm{sep}}(n,d,\varepsilon)
&=\max_{(\varepsilon_1,\ldots,\varepsilon_d):\sum_{i\in[d]}\varepsilon_i\leq \varepsilon}\min_{i\in[d]}\left\{\frac{nC}{d}+\sqrt{\frac{n\rmV_{\varepsilon_i}}{d}}\Phi^{-1}(\varepsilon_i)\right\}+O(\log n)\label{performance:seperate1}\\
&=\frac{nC}{d}+\sqrt{\frac{n\rmV_{\frac{\varepsilon}{d}}}{d}}\Phi^{-1}\left(\frac{\varepsilon}{d}\right)+O(\log n)\label{performance:seperate},
\end{align}
where \eqref{performance:seperate} follows since i) for any $\varepsilon\in(0,1)$ and any
$d\geq 2$, the minimization in \eqref{performance:seperate1} is achieved by some $i\in[d]$ such that $\varepsilon_i<0.5$ because $\Phi^{-1}(a)$ is decreasing in $a\in[0,1]$ and $\Phi^{-1}(a)< 0$ for any $a< 0.5$, and ii) for any $\varepsilon\in[0,1]$, the maximization is achieved by a vector $(\varepsilon_1,\ldots,\varepsilon_d)$ where $\varepsilon_i=\frac{\varepsilon}{d}$ for all $i\in[d]$.

Note that for any $\varepsilon\geq 0.5$ and any $d\geq 2$, the right hand side of \eqref{performance:seperate} is no greater than $\frac{nC}{d}+O(\log n)$ since $\Phi^{-1}(\frac{\varepsilon}{d})\leq \Phi^{-1}(0.5)=0$. However, the right hand side of \eqref{joint:search} is greater than $\frac{nC}{d}+O(\log n)$ since both $\Phi^{-1}(\varepsilon)$ and $V_\varepsilon$ are positive when $\varepsilon>0.5$. Furthermore, when $\varepsilon<0.5$, we have
\begin{align}
\frac{nC}{d}+\sqrt{\frac{n\rmV_{\frac{\varepsilon}{d}}}{d}}\Phi^{-1}\left(\frac{\varepsilon}{d}\right)
&=\frac{nC}{d}+\sqrt{\frac{n\rmV_{\varepsilon}}{d}}\Phi^{-1}\left(\frac{\varepsilon}{d}\right)\label{decreasePhi-1}\\
&<\frac{nC}{d}+\sqrt{\frac{n\rmV_{\varepsilon}}{d}}\Phi^{-1}\left(\varepsilon\right)\label{phi-1non}\\
&<\frac{nC}{d}+\sqrt{\frac{n\rmV_{\varepsilon}}{d^2}}\Phi^{-1}(\varepsilon)\label{usepositivePhiinv}\\
&=\frac{1}{d}\Big(nC+\sqrt{nV_\varepsilon}\Phi^{-1}(\varepsilon)\Big)
\end{align}
where \eqref{decreasePhi-1} follows since $V_\varepsilon$ (cf. \eqref{def:veps}) takes the same value for any $\varepsilon\in[0,0.5)$, \eqref{phi-1non} follows since $\Phi^{-1}(a)$
 in non-decreasing in $a\in[0,1]$ and $\varepsilon>\frac{\varepsilon}{2}\geq \frac{\varepsilon}{d}$, \eqref{usepositivePhiinv} follows since $d\geq 2$ and $\Phi^{-1}(\varepsilon)<0$ for any $\varepsilon<0.5$. Therefore, for any $d\geq 2$, the result in \eqref{performance:seperate} is always smaller than the result in \eqref{joint:search}. This implies that separate searching over each dimension of a multidimensional target variable is in fact \emph{not} optimal. This is verified by a numerical simulation in Section \ref{sec:numerical} (Figure \ref{sim_non_adap_sep}).

In the following, we specialize Theorem \ref{result:second} to different measurement-dependent channels.
\subsection{Case of Measurement-Dependent BSC}
 We first consider a measurement-dependent BSC. Given any $\nu\in(0,1]$ and any $q\in[0,1]$, let $\beta(\nu,q):=q(1-\nu q)+(1-q)\nu q$. For any $(x,y)\in\{0,1\}^2$, the information density of a measurement-dependent BSC with parameter $\nu$ is 
\begin{align}
\imath_{q,\nu q}(x;y)
\nn&=\bbo(x\neq y)\log(\nu q)+\bbo(x=y)\log(1-\nu q)-\bbo(y=1)\log(\beta(\nu,q))\\*
&\qquad-\bbo(y=0)\log(1-\beta(\nu,q)).
\end{align}
The mean and variance of the information density are respectively
\begin{align}
C(\nu,q)&:=\mathbb{E}[\imath_{q,\nu q}(X;Y)]=h_\rmb(\beta(\nu,q))-h_\rmb(\nu q)\label{def:Cnuq},\\
V(\nu,q)&:=\mathrm{Var}[\imath_{q,\nu q}(X;Y)],
\end{align}
where $h_\rmb(p)=-p\log(p)-(1-p)\log(1-p)$ is the binary entropy function. The capacity of the measurement-dependent BSC with parameter $\nu$ is thus
\begin{align}
C(\nu)=\max_{q\in[0,1]}C(\nu,q)\label{def:Cnu},
\end{align}
\begin{figure}[tb]
\centering
\includegraphics[width=.5\columnwidth]{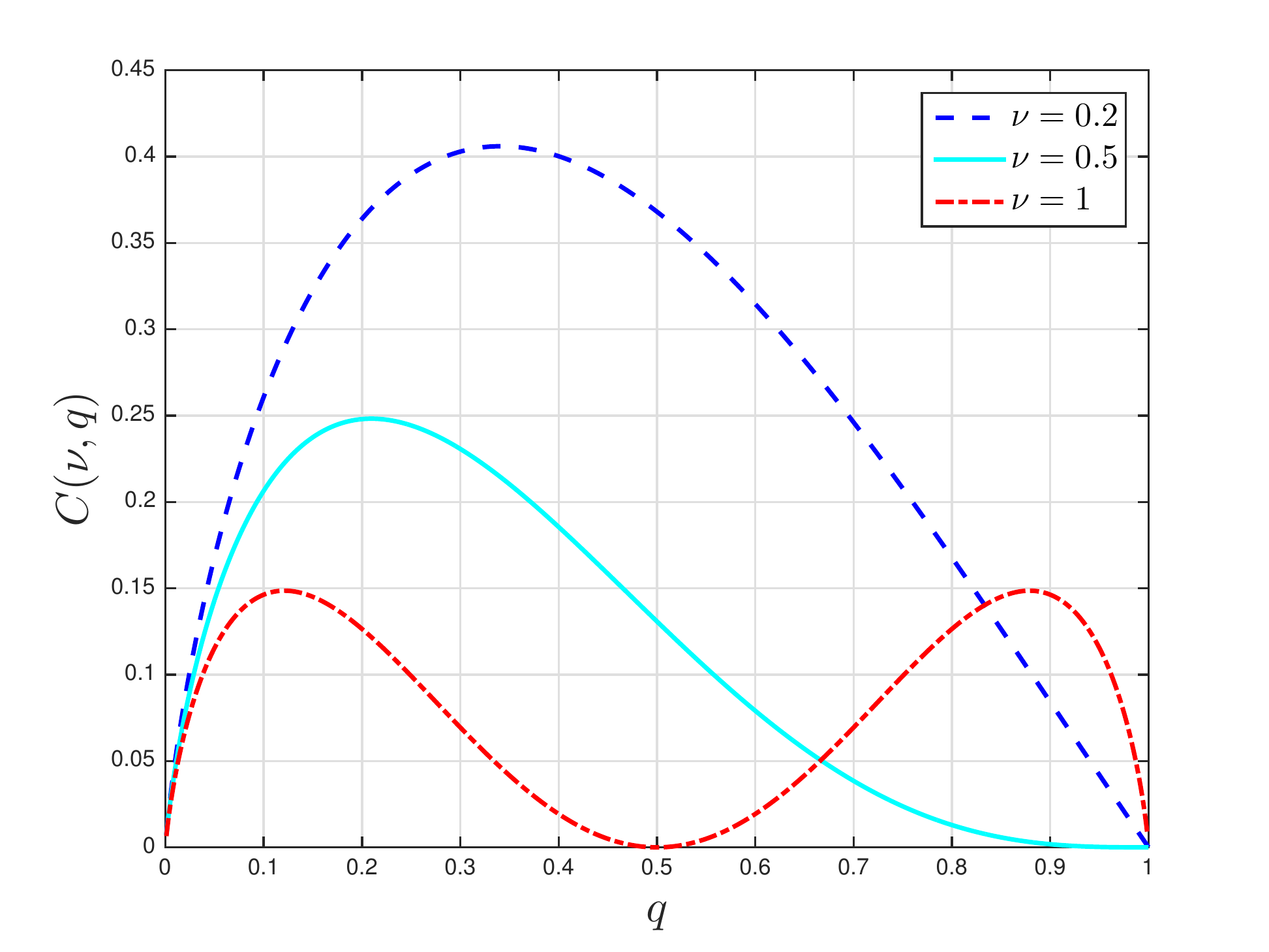}
\caption{Plot of $C(\nu,q)$, the mean of the mutual information density, of a measurement-dependent BSC for various values of $\nu$ and $q\in[0,1]$. For a given $\nu$, the maximum value of $C(\nu,q)$ over $q\in[0,1]$ is the capacity of the measurement-dependent BSC with parameter $\nu$ and the values of $q$ achieving this maximum consists of the set of capacity-achieving parameters $\calP_{\rm{ca}}$. For $\nu=0.2$ and $\nu=0.5$, $\calP_{\rm{ca}}$ is singleton and for $\nu=1$, $\calP_{\rm{ca}}$ contains two elements.}
\label{bsc_capacity}
\end{figure}

Depending on the value of $\nu\in(0,1]$, the set of capacity-achieving parameters $\calP_{\rm{ca}}$ may or may not be a singleton (cf. Figure \ref{bsc_capacity}). In particular, for any $\nu\in(0,1)$, the capacity-achieving parameter $q^*$ is unique. When $\nu=1$, there are two capacity-achieving parameters $q_1^*$ and $q_2^*$ where $q_1^*+q_2^*=1$. It can be verified easily that $V(1,q_1^*)=V(1,1-q_1^*)$. As a result, for any capacity-achieving parameter $q^*$ of the measurement-dependent BSC with parameter $\nu\in(0,1]$, the dispersion of the channel is
\begin{align}
V(\nu)&=V(\nu,q^*).
\end{align}

\begin{corollary}
\label{bsc:non-adaptive}
Let $\nu\in(0,1)$. If the channel from the oracle to the player is a measurement-dependent BSC with parameter $\nu$, then Theorem \ref{result:second} holds with $C=C(\nu)$ and $V_\varepsilon=V(\nu)$ for any $\varepsilon\in(0,1)$.
\end{corollary}
We make the following observations. Firstly, if we let $\nu=1$ and take $n\to\infty$, then for any $\varepsilon\in(0,1)$,
\begin{align}
\lim_{n\to\infty}\frac{-\log\delta^*(n,d,\varepsilon)}{n}=\frac{\max_{q\in[0,1]} \big(h_\rmb(\beta(1,q))-h_\rmb(q)\big)}{d}.
\end{align}
This is a strengthened version of \cite[Theorem 1]{kaspi2018searching} with strong converse.

Secondly, when one considers the measurement-independent BSC with parameter $\nu\in(0,1)$, then it can be shown that the achievable resolution $\delta^*_{\rm{mi}}(n,d,\varepsilon)$ of an optimal non-adaptive query procedure satisfies
\begin{align}
-d\log_2 \delta^*_{\rm{mi}}(n,d,\varepsilon)=n(1-h_\rmb(\nu))+\sqrt{n\nu(1-\nu)}\log_2\frac{1-\nu}{\nu}\Phi^{-1}(\varepsilon)+O(\log n).
\end{align}

\subsection{Case of Measurement-Dependent BEC}

We next consider a measurement-dependent BEC. Given any $\tau\in[0,1]$ and any $q\in(0,1)$, for any $(x,y)\in\{0,1\}\times\{0,1,\rme\}$, the information density for a measurement-dependent BEC with parameter $\tau$ is
\begin{align}
\imath_{q,q \tau}(x;y)
&=\bbo(y=x)\log(1-q\tau)-\bbo(y=1)\log(q(1-q\tau))-\bbo(y=0)\log((1-q)(1-q\tau)).
\end{align}
The mean and variance of the information density are respectively
\begin{align}
C(\tau,q)
&:=\mathbb{E}[\imath_{q,q \tau}(X;Y)]=(1-q\tau)h_\rmb(q),\\
V(\tau,q)
\nn&:=\mathrm{Var}[\imath_{q,q\tau}(X;Y)]=(1-q\tau)\Big(h_\rmb(q)\log(1-q\tau)+q\log q\log(q(1-q\tau))\\*
&\qquad+(1-q)\log(1-q)\log((1-q)(1-q\tau))\Big)-(1-q\tau)^2h_\rmb(q)^2.
\end{align}

The capacity of the measurement-dependent BEC with parameter $\tau\in[0,1]$ is given by
\begin{align}
C(\tau)=\max_{q\in[0,1]}C(\tau,q)=\max_{q\in[0,0.5]}(1-q\tau)h_\rmb(q),
\end{align}
where the second equality follows since for any $\tau$, $C(\tau,q)$ is decreasing in $q\in[0.5,1]$. We plot $C(\tau,q)$ for different values of $\tau$ in Figure \ref{bec_capacity}. It can be verified that the capacity-achieving parameter for the measurement-dependent BEC is unique and we denote it by $q^*$. Thus, the dispersion of the channel is $V(\tau,q^*)$.

\begin{figure}[tb]
\centering
\includegraphics[width=.5\columnwidth]{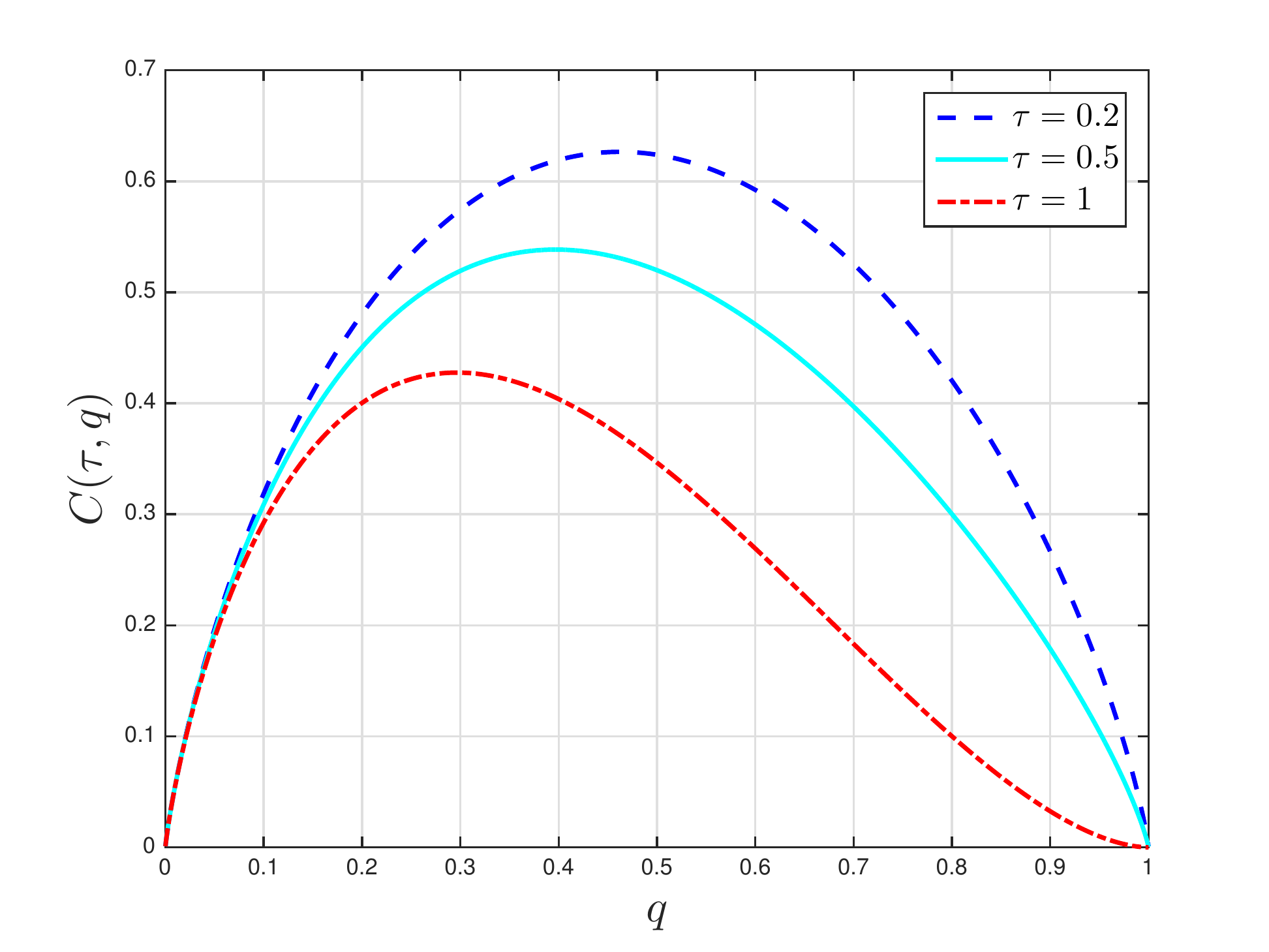}
\caption{Plot of $C(\tau,q)$ for the measurement-dependent BEC with parameter $\tau$ for $q\in[0,1]$. For each given $\tau$, the maximum value of $C(\tau,q)$ over $q\in[0,1]$ is the capacity of the measurement-dependent BEC. Note that the capacity-achieving parameter $q^*$ for measurement-dependent BEC is unique for any $\tau\in(0,1]$.}
\label{bec_capacity}
\end{figure}

We obtain the following:
\begin{corollary}
\label{bec:non-adaptive}
Let $\tau\in[0,1]$. If the channel from the oracle to the player is a measurement-dependent BEC with parameter $\tau$, then Theorem \ref{result:second} holds with $C=C(\tau)$ and $V_\varepsilon=V(\tau,q^*)$ for any $\varepsilon\in[0,1)$.
\end{corollary}
The remarks we made for Corollary \ref{bsc:non-adaptive} apply equally to Corollary \ref{bec:non-adaptive}, but additional properties are worthwhile to mention. Firstly, if one considers a measurement-independent BEC with parameter $\tau\in[0,1]$, then the achievable resolution $\delta^*_{\rm{mi}}(n,d,\varepsilon)$ of optimal non-adaptive query procedures satisfies
\begin{align}
-d\log_2 \delta^*_{\rm{mi}}(n,d,\varepsilon)=n(1-\tau)+\sqrt{n\tau(1-\tau)}\Phi^{-1}(\varepsilon)+O(\log n).
\end{align}

Secondly, if the channel is a noiseless (i.e., $\tau=0$), the achievable resolution of optimal non-adaptive query procedures satisfies
\begin{align}
-d\log_2 \delta^*(n,d,\varepsilon)=n+O(\log n).
\end{align}
Note that, interestingly, for the noiseless 20 questions problem, the achievable resolution of optimal non-adaptive querying does not depend on the target excess-resolution probability $\varepsilon\in[0,1)$ for any number of queries $n$. This is in contrast to the noisy 20 questions problem where a similar phenomenon occurs only when $n\to\infty$, c.f. \eqref{phasetransition}. The implication is that that in the noiseless 20 questions problem, for any number of the queries $n\in\bbN$, the achievable resolution of optimal non-adaptive query procedures cannot be improved even if one tolerates a larger excess-resolution probability $\varepsilon$.

\subsection{Case of Measurement-Dependent Z-Channel}
We next consider a measurement-dependent Z-channel. Given any $\zeta\in(0,1]$ and $q\in(0,1]$, for any $(x,y)\in\{0,1\}^2$, the information density of a measurement-dependent Z-channel with parameter $\zeta$ is
\begin{align}
\imath_{q,\zeta q}(x;y)
&=\bbo(y=x=0)\log\frac{1}{1-q+\zeta q^2}+\bbo(y=0,x=1)\log\frac{\zeta q}{1-q+\zeta q^2}+\bbo(y=x=1)\log\frac{1-\zeta q}{q-\zeta q^2}.
\end{align}
The mean and the variance of the information density are respectively
\begin{align}
C(\zeta,q)
&:=\mathbb{E}[\imath_{q,\zeta q}(X;Y)]
=h_\rmb(q(1-\zeta q))-q h_\rmb(\zeta q),\\
V(\zeta,q)
\nn&:=\mathrm{V}[\imath_{q,\zeta q}(X;Y)].
\end{align}
The capacity of the measurement-dependent Z-channel with parameter $\zeta$ is
\begin{align}
C(\zeta)
&=\max_{q\in[0,1]}C(\zeta,q).
\end{align}
We plot $C(\zeta,q)$ for different values of $\zeta$ and $q\in[0,1]$ in Figure \ref{z_capacity}.
\begin{figure}[tb]
\centering
\includegraphics[width=.5\columnwidth]{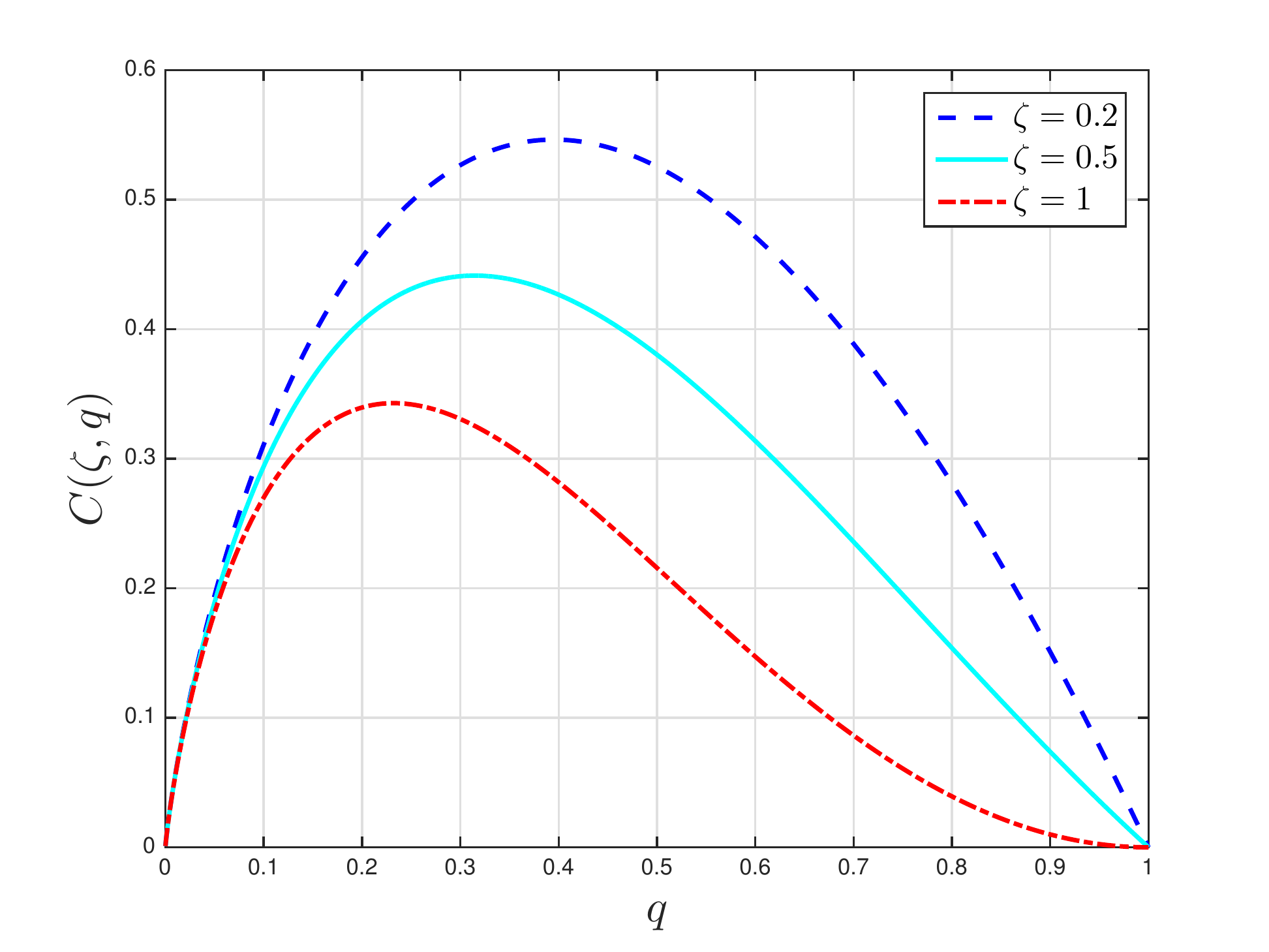}
\caption{Plot of $C(\zeta,q)$, the mean of the information density, of the measurement-dependent binary Z-channel with parameter $\zeta$ for $q\in[0,1]$. For any $\zeta\in(0,1]$, there exists a unique capacity-achieving value $q^*\in[0,1]$.}
\label{z_capacity}
\end{figure}
It can be verified that the capacity achievable parameter for the measurement-dependent Z-channel is unique and we denote the optimizer as $q^*$. Therefore, the dispersion of the Z-channel is $V(\zeta,q^*)$. Our second-order asymptotic result in Theorem \ref{result:second} specializes to the Z-channel as follows.
\begin{corollary}
\label{bz:non-adaptive}
Let $\zeta\in(0,1]$. If the channel from the oracle to the player is a measurement-dependent Z-channel with parameter $\zeta$, then Theorem \ref{result:second} holds with $C=C(\zeta,q^*)$ and $V_\varepsilon=V(\zeta,q^*)$ for any $\varepsilon\in(0,1)$.
\end{corollary}

When one considers a measurement-independent Z-channel with parameter $\zeta$, it can be easily shown that the achievable resolution $\delta^*_{\rm{mi}}(n,d,\varepsilon)$ of optimal non-adaptive query procedures satisfies
\begin{align}
-d\log \delta_{\rm{mi}}^*(n,d,\varepsilon)=nC_{\rm{mi}}(\zeta)+\sqrt{nV_{\rm{mi}}(\zeta)}\Phi^{-1}(\varepsilon)+O(\log n),
\end{align}
where $C_{\rm{mi}}(\zeta)$ and $V_{\rm{mi}}(\zeta)$ are the capacity and dispersion of the Z-channel:
\begin{align}
C_{\rm{mi}}(\zeta)&=\sup_{q\in[0,1]}h_\rmb(q(1-\zeta))-qh_\rmb(\zeta),\\
V_{\rm{mi}}(\zeta)\nn&=q^*_{\rm{mi}}(1-q^*_{\rm{mi}})\Big(\log(1-q^*_{\rm{mi}}+\zeta q^*_{\rm{mi}})\Big)^2+\zeta q^*_{\rm{mi}}(1-\zeta q^*_{\rm{mi}})\Big(\log\frac{\zeta}{1-q^*_{\rm{mi}}+\zeta q^*_{\rm{mi}}}\Big)^2\\
&\qquad+(q^*_{\rm{mi}}-\zeta q^*_{\rm{mi}})(1-q^*_{\rm{mi}}+\zeta q^*_{\rm{mi}})\Big(\log\frac{1-\zeta}{q-\zeta q^*_{\rm{mi}}}\Big),
\end{align}
with $q^*_{\rm{mi}}\in[0,1]$ being the unique optimizer of $C_{\rm{mi}}(\zeta)$.

\section{Upper Bound on Resolution of Adaptive Querying}
\label{sec:ach}
In this section, we present a second-order asymptotic upper bound on the achievable resolution of adaptive query procedures and use this bound to discuss the benefit of adaptivity. We remark that the adaptive query procedure used in our proof is a special case of general adaptive querying in the sense that our query sets are designed in a non-adaptive manner and the stopping time varies as a function of noisy responses. Such a design is based on the achievable coding scheme for the variable length feedback code in \cite{polyanskiy2011feedback} for channel coding with feedback.

Recall the definition of the capacity $C$ of measurement-dependent channels in \eqref{def:capacity}.
\begin{theorem}
\label{second:fbl:adaptive}
For any $(l,d,\varepsilon)\in\bbR_+\times\bbN\times[0,1)$,
\begin{align}
-\log\delta^*_\rma(l,d,\varepsilon)\geq \frac{lC}{d(1-\varepsilon)}+O(\log l)\label{md:adaptive}.
\end{align}
\end{theorem}
The proof of Theorem \ref{second:fbl:adaptive} is in Appendix \ref{proof:second:fbl:adaptive}. 

We make several remarks. A converse bound is necessary to establish the optimality of any adaptive query algorithm under a measurement-dependent channel. However, a converse is elusive, since as pointed out in \cite{kaspi2018searching}, under the measurement-dependent channel, each noisy response $Y_i$ depends not only on the target vector $\bS$, but also the previous queries $\calA^{i-1}$ and noisy responses $Y^{i-1}$. This strong dependency makes it difficult to directly relate the current problem to channel coding with feedback~\cite{horstein1963sequential}. Indeed, under such a setting, the corresponding classical coding analogy is channel coding with feedback and with state where the state has memory. New ideas and techniques are likely required to establish a converse proof for this setting.

Secondly, by comparing Theorem \ref{result:second} to Theorem \ref{second:fbl:adaptive}, we can analyze the benefit of adaptivity for the noisy 20 questions problem with measurement-dependent noise. For any $(n,d,\varepsilon)\in\bbN^2\times[0,1)$, define the benefit of adaptivity, called adaptivity gain, as
\begin{align}
\rmG(n,d,\varepsilon):=\log \delta^*(n,d,\varepsilon)-\log\delta^*_\rma(n,d,\varepsilon).
\end{align}
Using Theorems \ref{result:second} and \ref{second:fbl:adaptive}, we have
\begin{align}
\rmG(n,d,\varepsilon)
\geq\frac{1}{d}\bigg(\frac{nC\varepsilon}{1-\varepsilon}-\sqrt{nV_\varepsilon}\Phi^{-1}(\varepsilon)\bigg)+O(\log n)=:\underline{\rmG}(n,d,\varepsilon).
\end{align}
Note that $\Phi^{-1}(\varepsilon)<0$. To illustrate the adaptivity again, Figure \ref{gain_adaptivity}, we plot $\underline{\rmG}(n,d,\varepsilon)$ for $d=2$, $\varepsilon=0.001$ and three types of measurement-dependent channels with various parameters. Note that adaptive query procedures gain over non-adaptive query procedures since for the former, one can make different number of queries with respect to different realization of the target variable.
\begin{figure}[tb]
\centering
\begin{tabular}{ccc}
\hspace{-.25in} \includegraphics[width=.33\columnwidth]{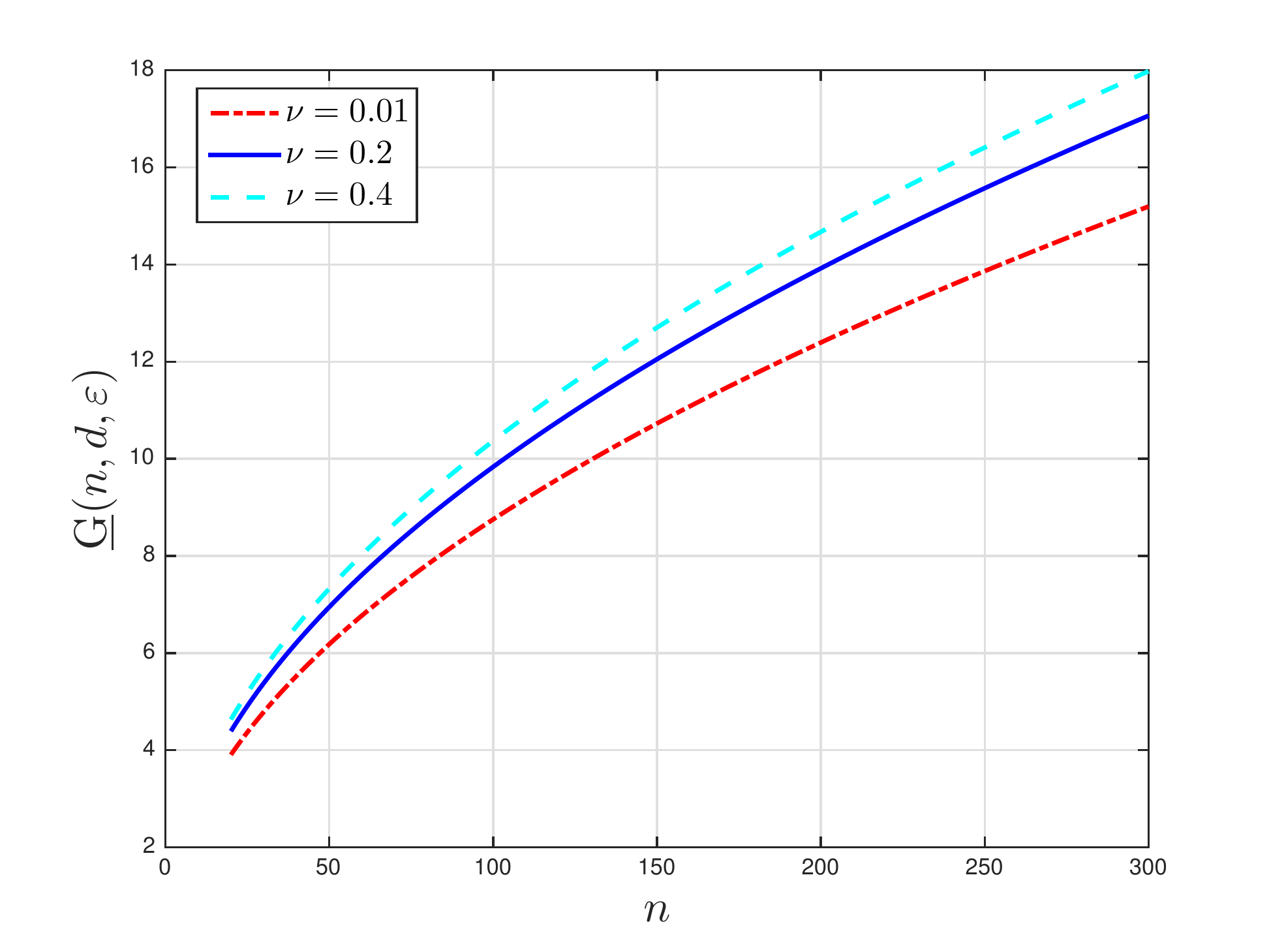}& \hspace{-.4in} \includegraphics[width=.33\columnwidth]{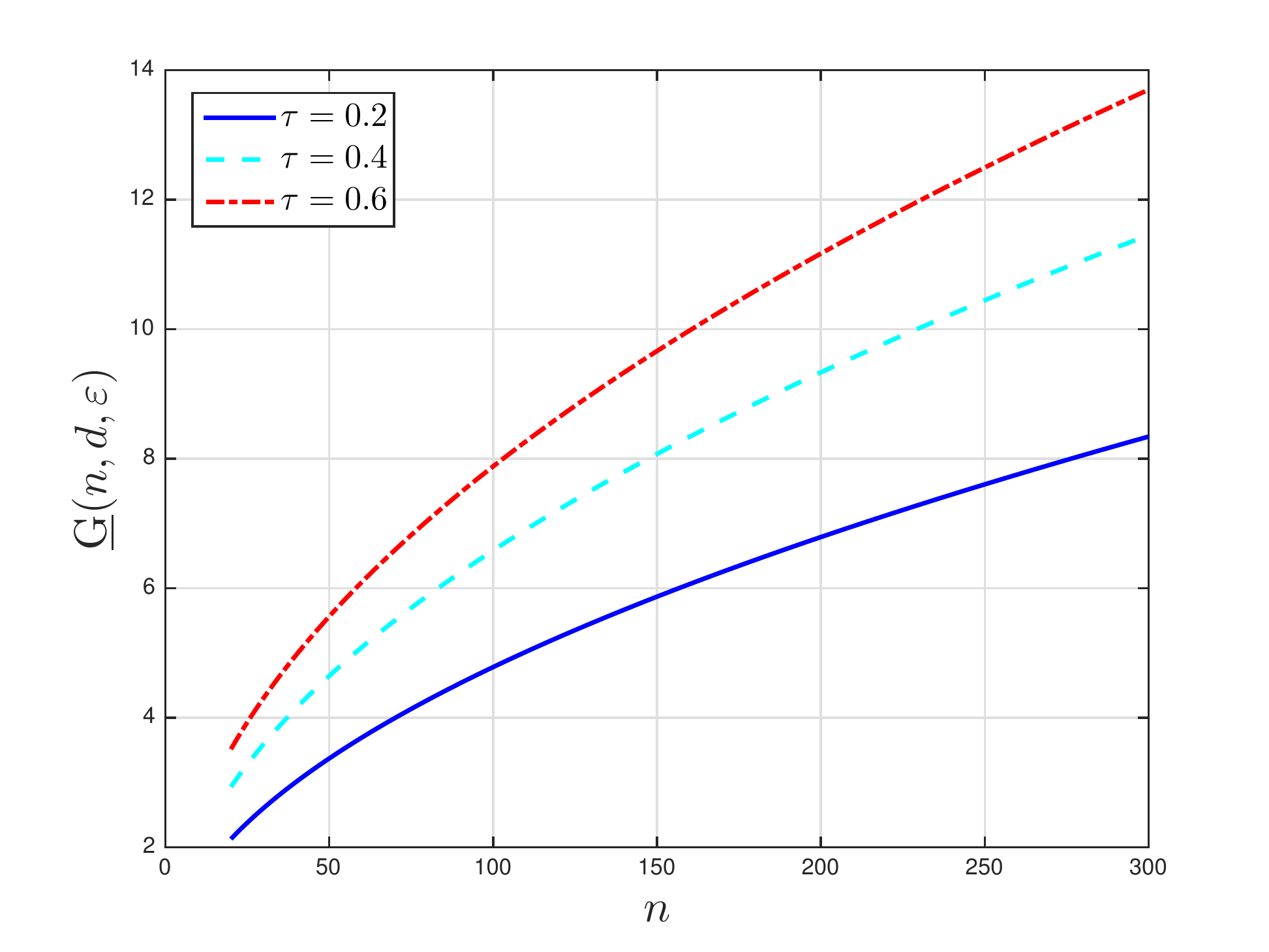}&\hspace{-.4in} \includegraphics[width=.33\columnwidth]{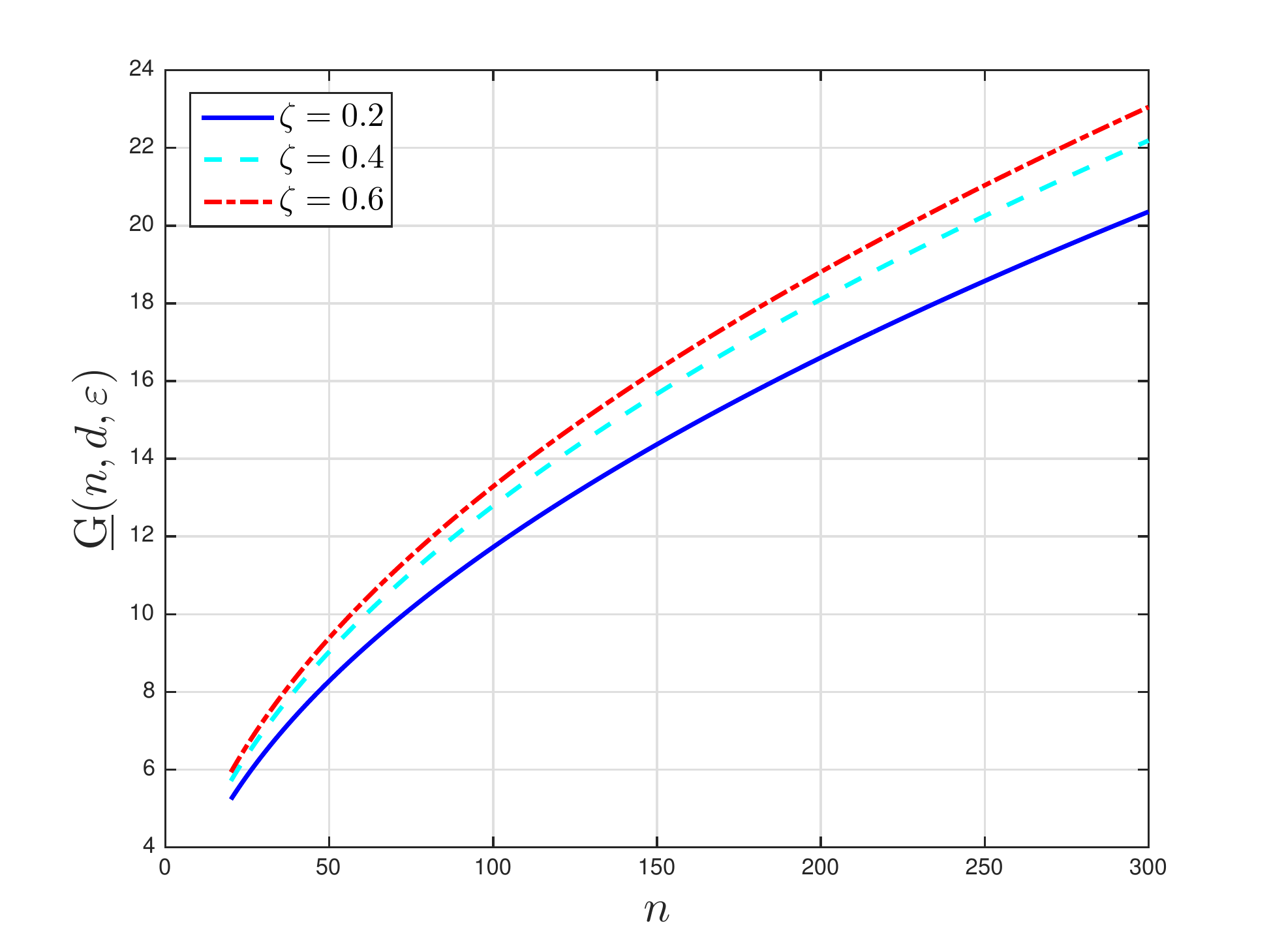}\\
\hspace{-.25in} {(a) measurement-dependent BSC} & \hspace{-.4in}  { (b) measurement-dependent BEC}& \hspace{-.4in}  { (c) measurement-dependent Z-channel}\\
\end{tabular}
\caption{Lower bound $\underline{\rmG}(n,d,\varepsilon)$ on the benefit of adaptivity  where $\underline{\rmG}(n,d,\varepsilon)=\frac{1}{d}\Big(\frac{nC\varepsilon}{1-\varepsilon}-\sqrt{nV_\varepsilon}\Phi^{-1}(\varepsilon)\Big)+O(\log n)$. The $O(\log n)$ term is not included in the plots. We consider the case of $d=2$ and $\varepsilon=0.001$.}
\label{gain_adaptivity}
\end{figure}

Finally, using the techniques in \cite{polyanskiy2011feedback} and the relationship between adaptive querying in 20 questions and channel coding with feedback, we have that the achievable resolution $\delta^*_{\rma,\rm{mi}}(l,d,\varepsilon)$ of optimal adaptive query procedures for for measurement-independent channels satisfies
\begin{align}
-\log\delta^*_{\rma,\rm{mi}}(l,d,\varepsilon)=\frac{lC_{\rm{mi}}}{d(1-\varepsilon)}+O(\log l),\label{mi:adaptive}
\end{align}
where $C_{\rm{mi}}$ is the capacity of the measurement-independent channel. 

\section{Numerical Illustrations}
\label{sec:numerical}

In this subsection, we present numerical simulations to illustrate Theorem \ref{result:second} on non-adaptive searching for a multidimensional target. We consider the case where the target variable $\bS=(S_1,\ldots,S_d)$ is uniformly distributed over the unit cube of dimension $d$. We consider a measurement-dependent BSC with parameter $\nu=0.4$ and set the target excess-resolution probability to be $\varepsilon=0.1$ in call cases.

\begin{figure}[tb]
\centering
\includegraphics[width=.5\columnwidth]{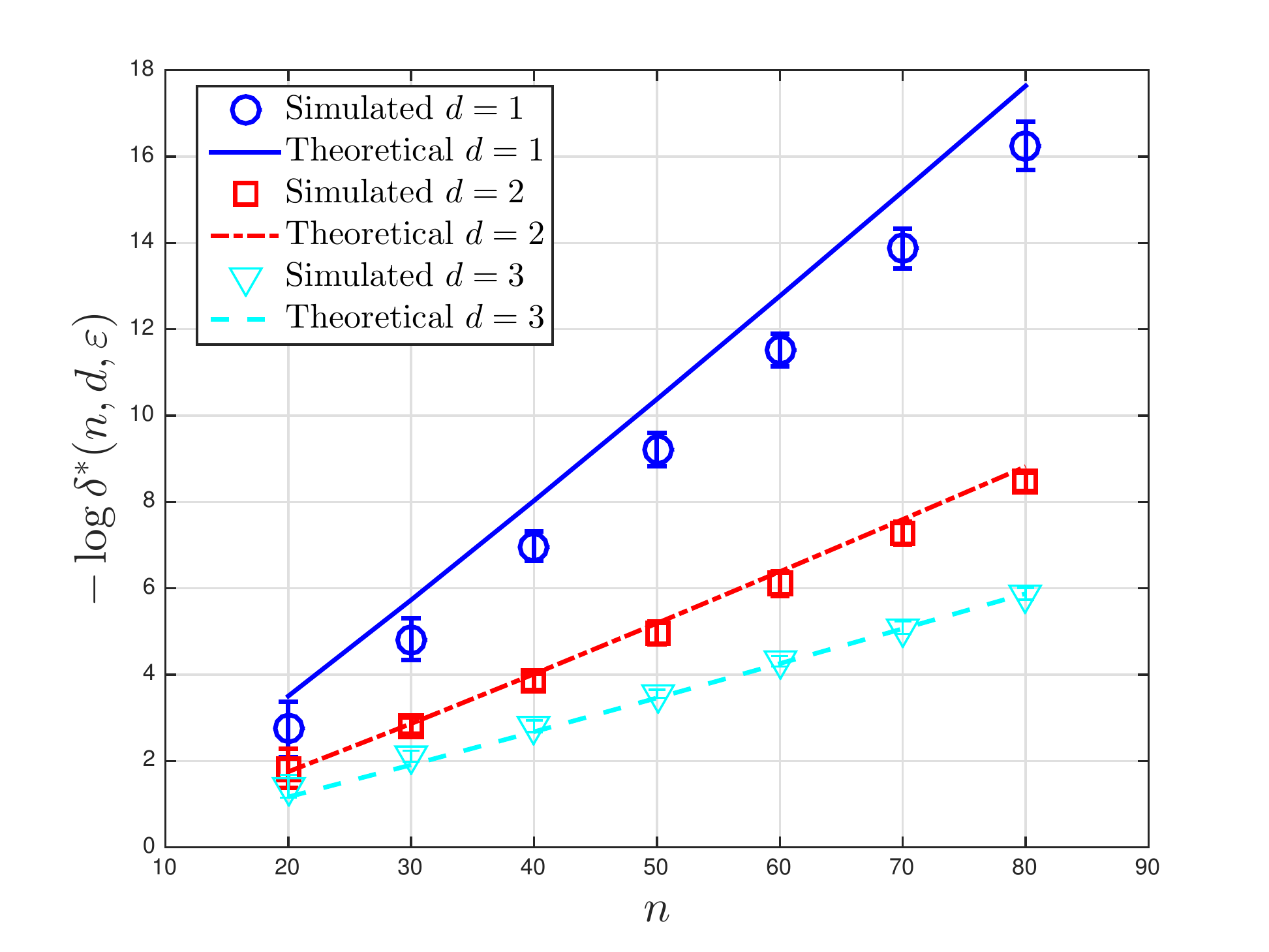}
\caption{Minimal achievable resolution of non-adaptive query procedures for estimating a uniformly distributed target variable $\bS=(S_1,\ldots,S_d)$ in the unit cube of dimension $d$. The theoretical results correspond to the second-order asymptotic result in Theorem \ref{result:second} and the simulate results correspond to the Monte Carlo simulation of the non-adaptive query procedure in Algorithm \ref{procedure:nonadapt}. The error bar for simulated results denotes thirty standard deviations below and above the mean.}
\label{sim_non_adap_ddim}
\end{figure}

In Figure \ref{sim_non_adap_ddim}, the simulated achievable resolution for the non-adaptive query procedure in Algorithm \ref{procedure:nonadapt} is plotted and compared to the theoretical predictions in Theorem \ref{result:second} for several values of the dimension $d$. Given $d\in\bbN$, for each $n\in\{20,30,\ldots,80\}$, the target resolution in the numerical simulation is chosen to be the reciprocal of $M$ such that
\begin{align}
\log M=\frac{1}{d}\left(nC(\nu)+\sqrt{nV(\nu)}\Phi^{-1}(\varepsilon)\right).
\end{align}
For each number of queries $n\in\{20,30,\ldots,80\}$, the non-adaptive query procedure in Algorithm \ref{procedure:nonadapt} is run independently $10^4$ times and the achievable resolution is calculated. From Figure \ref{sim_non_adap_ddim}, we observe that 
our theoretical result in Theorem \ref{result:second} provides a good approximation to the non-asymptotic performance of the query procedure in Algorithm \ref{procedure:nonadapt}.

\begin{figure}[tb]
\centering
\includegraphics[width=.5\columnwidth]{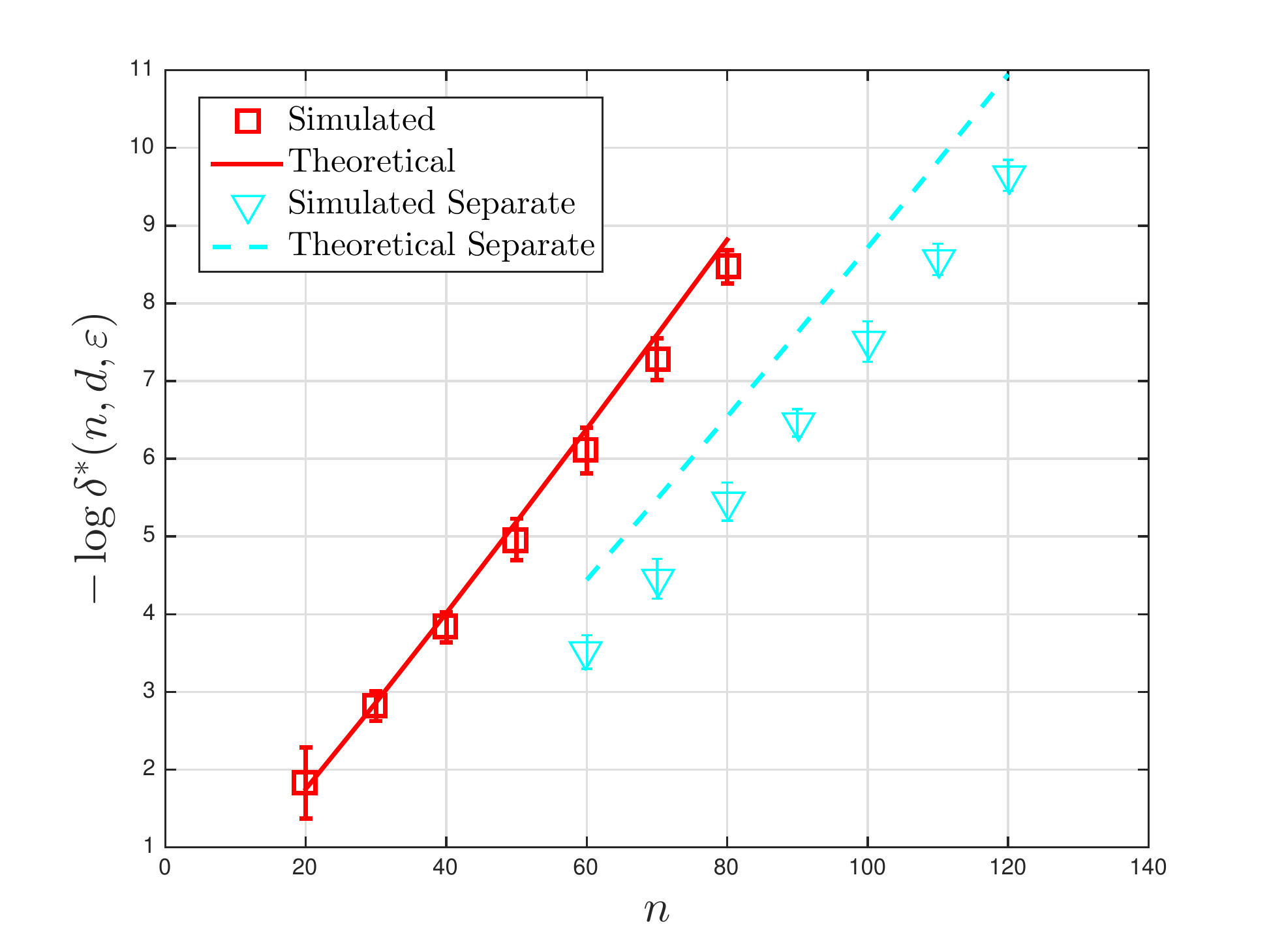}
\caption{
Minimal achievable resolution of non-adaptive query procedures of searching for a uniformly distributed target variable $\bS=(S_1,S_2)$ over the unit cube of dimension $d=2$. The red line corresponds to the second-order asymptotic result in Theorem \ref{result:second} and the red square denotes the Monte Carlo simulation of the non-adaptive query procedure in Algorithm \ref{procedure:nonadapt}. The cyan dashed line and triangle correspond to the second-order asymptotic result in \eqref{performance:seperate} and the Monte Carlo simulation of Algorithm \ref{procedure:nonadapt} for separate searching over each dimension of $\bS$ respectively. The error bar for the simulated results denotes thirty standard deviations above and below the mean.
\label{sim_non_adap_sep}
}
\end{figure}
In Figure \ref{sim_non_adap_sep}, for a $2$-dimensional target variable $\bS=(S_1,S_2)$, the simulated achievable resolutions is plotted for Algorithm \ref{procedure:nonadapt} and
a decoupled dimension-by-dimension search. Also shown are the theoretical predictions in Theorem \ref{result:second} and \eqref{performance:seperate} respectively. The gap between theoretical and simulated results arises since we have not accounted for the third-order term, which scales as $O(\log n)$. From Figure \ref{sim_non_adap_sep}, it can be observed that separate searching over each dimension is strictly suboptimal.

\section{Conclusion}
We derived the minimal achievable resolution of non-adaptive query procedures for the noisy 20 questions problem where the channel from the oracle to the player is a measurement-dependent discrete channel. Furthermore, we generalized our results to derive bounds on the achievable resolution of adaptive query procedures and discussed the intrinsic resolution benefit due to adaptivity.

There are several avenues for future research. Firstly, for adaptive query procedures, we derived achievability results on searching for a single multidimensional target over the unit cube. It would be fruitful to apply novel techniques to derive a matching converse bound on the minimal achievable resolution of optimal adaptive query procedures. It is also of interest to comprehensively compare the bound in Theorem \ref{second:fbl:adaptive} to the performance of  state-of-the-art adaptive query procedures, such as the sorted posterior matching algorithm~\cite{chiu2016sequential}. Secondly, we considered discrete channel (finite output alphabet). It would be interesting to extend our results to continuous channels such as a measurement-dependent additive white Gaussian noise channel~\cite{lalitha2018improved}. Thirdly, in this paper, we were interested in fundamental limits of optimal query procedures. One can propose low-complexity practical query procedures and compare the performances of their proposed query procedures to our derived benchmarks. Finally, we considered a single stationary target in the paper. In future work, it would be worthwhile to apply our second-order analysis to searching for a moving target with unknown velocity~\cite[Theorem 3]{kaspi2018searching} and to simultaneous searching for multiple targets~\cite{6970834}.

\appendix
\subsection{Proof of the Non-Asymptotic Achievability Bound (Theorem \ref{ach:fbl})}
\label{proof:ach}
In this subsection, we analyze the performance of the non-adaptive query procedure in Algorithm \ref{procedure:nonadapt} using ideas from channel coding~\cite{shannon1948mathematical}. To begin with, we first briefly recall the query procedure in Algorithm \ref{procedure:nonadapt}.

Fix any $M\in\bbN$, we partition the unit cube of dimension $d$ into $M^d$ equal-sized disjoint cubes $\{\calS_{i_1,\ldots,i_d}\}_{(i_1,\ldots,i_d)\in[M]^d}$. Let $\bx=\{x^n(i_1,\ldots,i_d)\}_{(i_1,\ldots,i_d)\in[M]^d}$ be a sequence of $M^d$ binary codewords. For each $t\in[n]$, the $t$-th query is designed as
\begin{align}
\calA_t
&:=\bigcup_{(i_1,\ldots,i_d)\in[M]^d:x_t(i_1,\ldots,i_d)=1}\calS_{i_1,\ldots,i_d}\label{def:query:ddim},
\end{align}
where $x_t(i_1,\ldots,i_d)$ denotes the $t$-th element of the codeword $x^n(i_1,\ldots,i_d)$. By the above query design, our $t$-th query to the oracle is whether the target $\bs=(s_1,\ldots,s_d)$ lies in the union of cubes with indices of the codewords whose $t$-th element are one. Hence, for each $t\in[n]$, the $t$-th element of each codeword can be understood as an indicator function for whether a particular cube would be queried in $i$-th question, with one being positive and zero being negative.

For subsequent analysis, given any $s\in[0,1]$, define the following quantization function
\begin{align}
\rmq(s):=\lceil sM\rceil\label{def:qs},
\end{align} 
Given any d-dimensional target variable $\bs$, we use $\bw=(w_1,\ldots,w_d)$ to denote the vector $(\rmq(s_1),\ldots,\rmq(s_d))$, i.e., $w_i=\rmq(s_i)$. Given $\bs$, the noiseless answer of the oracle to the query $\calA_t$ is
\begin{align}
Z_t
&=\bbo(\bs\in\calA_t)=\bbo\bigg(\bs\in\bigcup_{(i_1,\ldots,i_d)\in[M]^d:x_t(i_1,\ldots,i_d)=1}\calS_{i_1,\ldots,i_d}\bigg)\\
&=\bbo(x_t(\bw)=1)=x_t(\bw).
\end{align}
Then the noisy response $Y_t$ is obtained by passing $x_t(\rmq(\bs))$ over the measurement-dependent channel.

Given noisy responses $Y^n=(Y_1,\ldots,Y_n)$, the decoder produces estimates $\hatS=(\hatS_1,\ldots,\hatS_d)$ using the following two-step decoding:
\begin{enumerate}
\item the player first estimates $\bw$ as $\hat{\bW}=(\hatW_1,\ldots,\hatW_d)$ using a maximal mutual information decoder, i.e.,
\begin{align}
\hat{\bW}=(\hatW_1,\ldots,\hatW_d)=\max_{(\tili_1,\ldots,\tili_d)\in[M]^d}\imath_p(x^n(\tili_1,\ldots,\tili_d);Y^n);
\end{align}
\item the player then produces estimates $\hatS=(\hatS_1,\ldots,\hatS_d)$ as follows:
\begin{align}
\hatS_j=\frac{2\hatW_j-1}{2M}
\end{align}
for all $j\in[d]$.
\end{enumerate}

It is easy to verify that using the above query procedure, the estimate $\hatS_i$ is within $\frac{1}{M}$ of the target $s_i$ for all $i\in[d]$ if our estimate $\hat{\bW}=\bw$. Thus the excess-resolution probability of the multidimensional estimation is upper bounded by the error probability of channel coding with $M^d$ messages over a measurement-dependent codebook.

For subsequent analysis, we use $\bW=(W_1,\ldots,W_d)$ to denote the quantized vector of a target vector $\bS=(S_1,\ldots,S_d)\in[0,1]^d$, i.e., $W_i=\rmq(S_i)$ for each $i\in[d]$. We use $\bw$ to denote a particular realization. Note that each pdf $f_{\bS}\in\calF([0,1]^d)$ of the target vector $\bS$ induces a pmf $P_{\bW}\in\calP([M]^d)$. Using our query procedure, we have
\begin{align}
\nn&\sup_{f_{\bS}\in\calF([0,1]^d)}\Pr\left\{\exists~i\in[d],~|\hatS_i-S|>\frac{1}{M}\right\}\\*
&\leq \sup_{f_{\bS}\in\calF([0,1]^d)}\Pr\{\hat{\bW}\neq \bW\}\\
&\leq \sup_{P_{\bW}\in\calP([M]^d)}\Pr\{\hat{\bW}\neq \bW\}\\
&\leq \sup_{P_{\bW}\in\calP([M]^d)}\sum_{\bw}P_{\bW}(\bw)\Pr\{\exists~\bar{\bw}\in[M]^d:~\bar{\bw}\neq \bw,~\imath_p(x^n(\bar{\bw});Y^n)\geq \imath_p(x^n(\bw);Y^n)\}\label{specifydist}\\
&=:\sup_{P_{\bW}\in\calP([M]^d)}\sum_{\bw}P_{\bW}(\bw)\rmP_\rme(\bx,P_{\bW})\label{uppexcessp},
\end{align}
where the probability in \eqref{specifydist} is calculated with respect to the measurement-dependent channel 
\begin{align}
P_{Y^n|X^n}^{\calA^n}(y^n|x^n(\bw))
&=\prod_{t\in[n]}P_{Y|X}^{\calA_t}(y_t|x_t(\bw))\\
&=\prod_{t\in[n]}P_{Y|X}^{q_{t,d}^M(\bx)}(y_t|x_t(\bw))\label{useai},
\end{align}
and in \eqref{useai}, we define
\begin{align}
q_{t,d}^M(\bx):=\frac{1}{M^d}\sum_{\bw\in[M]^d}x_t(\bw).
\end{align}

Note that $\rmP_\rme(\bx,P_{\bW})$ is essentially the error probability of transmitting a message $\bW\in[M]^d$ with pmf $P_{\bW}$ over the measurement-dependent channel $P_{Y^n|X^n}^{\calA^n}$. Thus, to further bound $\rmP_\rme(\bx,P_{\bW})$, we need to analyze the error probability of a channel coding problem over a \emph{codebook dependent channel} where the channel output $Y^n$ depends on all codewords $\{x^n(i_1,\ldots,i_d)\}_{(i_1,\ldots,i_d)\in[M]^d}$. In contrast, in the classical channel coding problem, the channel output depends only on the channel input with respect to the message. However, as we shall see, using the change-of-measure technique and the assumption in \eqref{assump:continuouschannel}, with negligible loss in error probability, we can replace the measurement-dependent channel with a memoryless channel $(P_{Y|X}^p)^n$.

For this purpose, we use random coding ideas~\cite{gallager_ensemble}. Fix a Bernoulli distribution $P_X\in\calP(\{0,1\})$ with parameter $p$, i.e., $P_X(1)=p$. Let $\bX:=\{X^n(i_1,\ldots,i_d)\}_{(i_1,\ldots,i_d)\in[M]^d}$ be $M^d$ independent binary sequences, each generated i.i.d. from $P_X$. Furthermore, for any $(M,d,p,\eta)\in\bbN^2\times(0,1)\times\bbR_+$, define the following typical set of binary codewords $\bx$:
\begin{align}
\calT^n(M,d,p,\eta)
&:=\bigg\{\bx=\{x^n(i_1,\ldots,i_d)\}_{(i_1,\ldots,i_d)\in[M]^d}\in\calX^{Mdn}:\left|q_{t,d}^M(\bx)-p\right|\leq \eta,~\forall~t\in[n]\bigg\}\label{def:typical}.
\end{align}
For any $\bx\in\calT^n(M,d,p,\eta)$, recalling the query design in \eqref{def:query:ddim} and the condition in \eqref{assump:continuouschannel}, we have
\begin{align}
\log \frac{P_{Y^n|X^n}^{\calA^n}(y^n|x^n)}{(P_{Y|X}^p)^n(y^n|x^n)}
&=\sum_{t\in[n]}\log\frac{P_{Y|X}^{q_{t,d}^M(\bx)}(y_i|x_i)}{P_{Y|X}^p(y_i|x_i)}\leq n\eta c(p)\label{fromassumption}.
\end{align}

Note that given any $\bw\in[M]^d$, the joint distribution of $(\bX,Y^n)$ under the current query procedure is
\begin{align}
P_{\bX Y^n}^{\rm{md},\bw}(\bx,y^n)
&=\Big(\prod_{\bar{\bw}\in[M]^d}P_X^n(x^n(\bar{\bw}))\Big)\Big(\prod_{t\in[n]}P_{Y|X}^{q_{t,d}^M(\bx)}(y_t|x_t(\bw))\Big)\label{truedis}.
\end{align}
and furthermore, we need the following alternative joint distribution of $(\bX,Y^n)$ to apply the change-of-measure idea
\begin{align}
P_{\bX Y^n}^{p,\bw}(\bx,y^n)
&=\Big(\prod_{\bar{\bw}\in[M]^d}P_X^n(x^n(\bar{\bw}))\Big)\Big(\prod_{t\in[n]}P_{Y|X}^p(y_t|x_t(\bw))\Big)\label{altdis}.
\end{align}

For any message distribution $P_{\bW}\in\calP([M]^d)$,
\begin{align}
\mathbb{E}_{\bX}[\rmP_\rme(\bX,P_{\bW})]
&\leq \Pr\{\bX\notin\calT^n(M,d,p,\eta)\}+\mathbb{E}_{\bX}[\rmP_\rme(\bX,P_{\bW})\bbo(\bX\in\calT^n(M,d,p,\eta))]\\
&\leq 4n\exp(-2M^d\eta^2)+\mathbb{E}_{\bX}[\rmP_\rme(\bX,P_{\bW})\bbo(\bX\in\calT^n(M,d,p,\eta))]\label{useatypical},
\end{align}
where \eqref{useatypical} follows from \cite[Lemma 22]{tan2014state}, which provides an upper bound on the probability of the atypicality of i.i.d. random variables and implies that
\begin{align}
\Pr\{\bX\notin\calT^n(M,d,p,\eta)\}\leq 4n\exp(-2M^d\eta^2).
\end{align}

The second term in \eqref{useatypical} can be further upper bounded as follows:
\begin{align}
\nn&\mathbb{E}_{\bX}[\rmP_\rme(\bX,P_{\bW})\bbo(\bX\in\calT^n(M,d,p,\eta))]\\*
&=\sum_{\bw}P_{\bW}(\bw)\mathbb{E}_{P_{\bX Y^n}^{\rm{md},\bw}}[\bbo(\bX\in\calT^n(M,d,p,\eta))\bbo(\exists~\bar{\bw}\in[M]^d:~\bar{\bw}\neq\bw,~\imath_p(X^n(\bar{\bw});Y^n)\geq \imath_p(X^n(\bw);Y^n))]\\
&\leq \exp(n\eta c(p))\sum_{\bw}P_{\bW}(\bw)\Pr_{P_{\bX Y^n}^{p,\bw}}\{\exists~\bar{\bw}\in[M]^d:~\bar{\bw}\neq \bw,~\imath_p(X^n(\bar{\bw});Y^n)\geq \imath_p(X^n(\bw);Y^n)\}\label{cofmeasure}\\
&\leq \exp(n\eta c(p))\sum_{\bw}P_{\bW}(\bw)\sum_{\bar{\bw}\in[M]^d:\bar{\bw}\neq \bw}\Pr_{P_{\bX Y^n}^{p,\bw}}\{\imath_p(X^n(\bar{\bw});Y^n)\geq \imath_p(X^n(\bw);Y^n)\}\label{useunbound}\\
&=\exp(n\eta c(p))\mathbb{E}_{P_{X^nY^n}}[\min\{1,M^d\Pr_{P_X^n}\{\imath_p(\barX^n;Y^n)\geq \imath_p(X^n;Y^n)|X^n,Y^n\}]\}\label{rcu},
\end{align}
where \eqref{cofmeasure} follows from \eqref{fromassumption} and the change of measure technique, \eqref{useunbound} follows from the union bound, \eqref{rcu} follows by noting that the codewords $\{X^n(i_1,\ldots,i_d)\}_{(i_1,\ldots,i_d)\in[M]^d}$ are independent under $P_{\bX Y^n}^{\rm{alt}}$, the total number of codewords is no greater than $M^d$ and by applying ideas leading to the random coding union bound~\cite{polyanskiy2010finite}. In \eqref{rcu}, the joint distribution of $(X^n,Y^n)$ is
\begin{align}
P_{X^nY^n}(x^n,y^n)=\prod_{t\in[n]}P_X(x_t)P_{Y|X}^{p}(y_t|x_t).
\end{align}

Combining \eqref{useatypical} and \eqref{rcu}, we conclude that there exists a sequence of binary codewords $\bx$ such that $\rmP_\rme(\bx,P_{\bW})$ is upper bounded by the desired quantity for all message distributions $P_{\bW}\in\calP([M]^d)$ and thus the proof of Theorem \ref{ach:fbl} is completed.

\subsection{Proof of the Non-Asymptotic Converse Bound (Theorem \ref{fbl:converse})}
\label{proof:converse}

\subsubsection{Converse Proof}
Consider any sequence of non-adaptive queries $\calA^n\subseteq([0,1]^d)^n$ and any decoding function $g:\calY^n\to [0,1]^d$ such that the worst case excess-resolution probability with respect to a resolution $\delta$ is upper bounded by $\varepsilon$, i.e.,
\begin{align}
\sup_{f_{\bS}\in\calF([0,1]^d)}\Pr\big\{\exists~i\in[d]:~|\hatS_i-S_i|>\delta\big\}\leq\varepsilon\label{error4converse}.
\end{align}
As a result, for uniformly distributed target vector $\bS=(S_1,\ldots,S_d)$, the excess-resolution probability with respect to $\delta$ is also upper bounded by $\varepsilon$. In the rest of the proof, we consider a \emph{uniformly} distributed $d$-dimensional target $\bS$.

Let $\beta$ be any real number such that $\beta\leq\frac{1-\varepsilon}{2}\leq 0.5$ and let $\tilM:=\lfloor\frac{\beta}{\delta}\rfloor$. Define the following quantization function
\begin{align}
\rmq_\beta(s):=\lceil s\tilM\rceil,~\forall~s\in\calS\label{def:qbeta}.
\end{align}

Given any queries $\calA^n\in([0,1]^d)^n$, the noiseless responses from the oracle are $X^n=(X_1,\ldots,X_n)$ where for each $t\in[n]$, $X_t=\bbo(\bS\in\calA_t)$ is a Bernoulli random variable with parameter being the volume of $\calA_t$, which this follows from the definition of the measurement-dependent channel and the fact that the target variable $\bS$ is uniformly distributed. The noisy responses $Y^n$ is the output of passing $X^n$ over the measurement-dependent channel $P_{Y^n|X^n}^{\calA^n}$. Finally, an estimate $\hat{\bS}=(\hatS_1,\ldots,\hatS_d)$ is produced using the decoding function $g$.

For simplicity, let $\bW:=(W_1,\ldots,W_d)=(\rmq_\beta(S_1),\ldots,\rmq_\beta(S_d))$ and let $\hat{\bW}:=(\rmq_\beta(\hatS_1),\ldots,\rmq_\beta(\hatS_d))$. Similarly to \cite{kaspi2018searching}, we have that
\begin{align}
\Pr\{\hat{\bW}\neq \bW\}
&=\Pr\{\hat{\bW}\neq \bW,~\exists~i\in[d]:~|\hatS_i-S_i|>\delta\}+\Pr\{\hat{\bW}\neq \bW,~\forall~i\in[d]:~|\hatS_i-S_i|\leq\delta\}\\
&\leq \Pr\{\exists~i\in[d]:~|\hatS_i-S_i|>\delta\}+\Pr\{\hat{\bW}\neq \bW,~\forall~i\in[d]:~|\hatS_i-S_i|\leq\delta\}\\
&\leq \varepsilon+\Pr\{\hat{\bW}\neq \bW,~\forall~i\in[d]:~|\hatS_i-S_i|\leq\delta\}\label{useerror4converse}\\
&\leq \varepsilon+\Pr\{\exists~i\in[d]:~\hatW_i\neq W_i~\mathrm{and}~|\hatS_i-S_i|\leq\delta\}\\
&\leq \varepsilon+\sum_{i\in[d]}\Pr\{\hatW_i\neq W_i~\mathrm{and}~|\hatS_i-S_i|\leq\delta\}\\
&\leq \varepsilon+2d\delta\tilM\label{boundaryerror}\\
&\leq\varepsilon+2d\beta\label{use:tilM},
\end{align}
where \eqref{useerror4converse} follows from \eqref{error4converse}, \eqref{boundaryerror} follows since i) only when $S_i$ is within $\delta$ to the boundaries (left and right) of the sub-interval with indices $W_i=\rmq_{\beta}(S_i)$ can the events $\hatW_i\neq W_i$ and $|\hatS_i-S_i|\leq \delta$ occur simultaneously, ii) $S_i$ is uniformly distributed over $\calS$ and thus iii) the probability of the event $\{\hatW_i\neq W_i,~|\hatS_i-S_i|\leq \delta\}$ is upper bounded by $2\delta\tilM$, and \eqref{use:tilM} follows from the definition of $\tilM$. To ease understanding of the critical step \eqref{boundaryerror}, we have provided a figure illustration in Figure \ref{figureillus4converse}.

\begin{figure}[tb]
\centering
\setlength{\unitlength}{0.5cm}
\scalebox{1}{
\begin{picture}(20,3)
\linethickness{1pt}
\put(0,0.5){\makebox(0,0){...}}
\put(1,0){\line(1,0){18}}
\put(20,0.5){\makebox(0,0){...}}
\put(1,0){\line(0,1){1}}
\put(7,0){\line(0,1){1}}
\put(13,0){\line(0,1){1}}
\put(19,0){\line(0,1){1}}
\put(8,0){\line(0,1){0.5}}
\put(9,0){\line(0,1){0.5}}
\put(10,0){\line(0,1){0.5}}
\put(11,0){\line(0,1){0.5}}
\put(12,0){\line(0,1){0.5}}
\put(6,0){\line(0,1){0.5}}
\put(14,0){\line(0,1){0.5}}
\put(1.1,1.5){\makebox(0,0){$(k-2)\delta/\beta$}}
\put(7.1,1.5){\makebox(0,0){$(k-1)\delta/\beta$}}
\put(13,1.5){\makebox(0,0){$ k\delta/\beta$}}
\put(19.1,1.5){\makebox(0,0){$(k+1)\delta/\beta$}}
\put(7,0){\transparent{0.3}\color{green}{\rule{\unitlength}{1\unitlength}}}
\put(12,0){\transparent{0.3}\color{blue}{\rule{\unitlength}{1\unitlength}}}
\end{picture}
}
\caption{Figure illustration of \eqref{boundaryerror} for $S_1$. Let $\delta=\frac{1}{600}$ and $\beta=\frac{1}{6}$. Thus, we partition the unit interval $[0,1]$ into $\tilM=100$ sub-intervals each with length $\frac{1}{100}$. In the figure, we plot three consecutive sub-intervals with indices $(k-1,k,k+1)$ for some $k\in[2:\tilM-1]$. Note that the $k$-th interval starts from $\frac{(k-1)\delta}{\beta}$ and end at $\frac{k\delta}{\beta}$ and contains $\frac{1}{\beta}$ small intervals, each of length $\delta$. Suppose $S_1$ lies in $k$-th sub-interval, then only if $S_1$ is with $\delta=\frac{1}{600}$ of the boundaries in $k$-th sub-interval, denoted with shaded color, can we find $\hatS_1$ in adjacent sub-interval such that $|\hatS_1-S_1|\leq \delta$ and $\hatW_1=\rmq_{\beta}(\hatS_1)\neq\rmq_{\beta}(S_1)=W_1$.}
\label{figureillus4converse}
\end{figure}
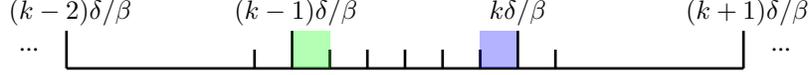

Using \eqref{use:tilM}, we have that the excess-resolution probability of searching for a multidimensional target variable is lower bounded by 
\begin{align}
\varepsilon
&\geq \Pr\{\hat{\bW}\neq \bW\}-2d\beta\label{conversechannel}.
\end{align}


Define the function $\tilde{\Gamma}:[\tilM]^d\to[\tilM^d]$ as
\begin{align}
\tilde{\Gamma}(i_1,\ldots,i_d)
&=1+\sum_{j\in[d]}i_j\tilM^{d-j}.
\end{align}
Using \eqref{conversechannel}, we have
\begin{align}
\varepsilon
&\geq \Pr\big\{\tilde{\Gamma}(\hat{\bW})\neq \tilde{\Gamma}(\bW)\big\}-2d\beta\label{usegammatocon}\\
&=\Pr\{\hatW\neq W\}-2d\beta\label{finalequation}
\end{align}
where \eqref{usegammatocon} follows since $\tilde{\Gamma}(\cdot)$ is a one-to-one mapping from $[\tilM]^d$ to $[\tilM^d]$ and in \eqref{finalequation}, we define $\hatW=\Gamma(\hat{\bW})\in[M^d]$ and define $W=\Gamma(\bW)\in[M^d]$ similarly. Note that from the problem formulation, since $\bS$ is uniformly distributed over $[0,1]^d$, we have that $\bW$ is uniformly distributed over $[\tilM]^d$ and thus $W$ is uniformly distributed over $[M^d]$.

Note that given queries $\calA^n$, the probability $\Pr\{\hatW\neq W\}$ is the average error probability of channel coding with deterministic states when the distribution of the channel inputs is $P_{X^n}^{\calA^n}$ and the message $W$ is uniformly distributed over $[\tilM^d]$. Therefore, we can use converse bounds for channel coding to bound achievable resolution $\delta$ (via $\tilM$).

Similar as \cite[Proposition 4.4]{TanBook} which provides a finite blocklength converse bound for the channel coding problem, we have that for any $\kappa\in(0,1-\varepsilon-2d\beta)$, 
\begin{align}
\log \tilM^d
&\leq\inf_{Q_{Y^n}\in\calP(\calY^n)}
\sup\bigg\{t\Big|\Pr\bigg\{\log\frac{P_{Y^n|X^n}^{\calA^n}(Y^n|X^n)}{Q_{Y^n}(Y^n)}\leq t\bigg\}\leq \varepsilon+2d\beta+\kappa\bigg\}-\log\kappa\label{fbl:converse:final}\\
&\leq\sup\bigg\{t\Big|\Pr\bigg\{\sum_{i\in[n]}\log\frac{P_{Y|X}^{\calA_i}(Y_i|X_i)}{P_Y^{|\calA_i|,|\calA_i|}(Y_i)}\leq t\bigg\}\leq \varepsilon+2d\beta+\kappa\bigg\}-\log\kappa\label{nconverse},
\end{align} 
where \eqref{nconverse} follows by choose $Q_Y^n$ being the marginal distribution of $Y^n$ induced distribution of $P_{X^n}^{\calA^n}$ and the measurement-dependent channel $P_{Y^n|X^n}^{\calA^n}$.  Note that \eqref{fbl:converse:final} is slightly different from \cite[Proposition 4.4]{TanBook}. In fact, we follow the proof of \cite[Proposition 4.4]{TanBook} with $M$ replaced by $\tilM^d$ and $\varepsilon$ replaced by $\varepsilon+2d\beta$ till the left hand side of  
\cite[Eq. (4.18)]{TanBook}. Then, we use the definition of the so called $\varepsilon$-hypothesis testing divergence~\cite[Eq. (2.9)]{TanBook}.

Since \eqref{nconverse} holds for any sequence of queries $\calA^n\in[0,1]^{nd}$ and any decoder $g:\calY^n\to[0,1]^d$ satisfying \eqref{error4converse}, recalling the definition of $\tilM$ and the definition of $\imath_{\calA_i}(\cdot)$, we have
\begin{align}
-d\log\delta\leq -d\log\beta-\log\kappa+\sup_{\calA^n\in[0,1]^{nd}}\sup\bigg\{t\in\bbR_+\Big|\Pr\Big\{\sum_{i\in[n]}\imath_{\calA_i}(X_i;Y_i)\leq t\bigg\}\leq \varepsilon+2d\beta+\kappa\Big\}.
\end{align}

\subsection{Proof of Second-Order Asymptotics (Theorem \ref{result:second})}
\label{proof:second}
\subsubsection{Achievability Proof}
Invoking Theorem \ref{ach:fbl} with the capacity-achieving parameter $q\in\calP_{\rm{ca}}$, we have that for any $\eta\in\bbR_+$, there exists a non-adaptive query procedure with $n$ queries such that
\begin{align}
\rmP_\rme^n\left(\frac{1}{M}\right)
&\leq 4n\exp(-2M^d\eta^2)+\exp(n\eta c(q))\mathbb{E}[\min\{1,M^d\Pr\{\imath_q(\barX^n;Y^n)\geq \imath_q(X^n;Y^n)\}\}]\label{step1}.
\end{align}
We first bound the expectation term in \eqref{step1} as follows:
\begin{align}
\nn&\mathbb{E}[\min\{1,M^d\Pr\{\imath_q(\barX^n;Y^n)\geq \imath_q(X^n;Y^n)\}\}]\\*
&\leq \Pr\left\{M^d\exp(-\imath_q(X^n;Y^n))\geq \frac{1}{\sqrt{n}}\right\}+\frac{1}{\sqrt{n}}\label{usemany}\\
&=\Pr\left\{d\log M-\imath_q(X^n;Y^n)\geq -\log \sqrt{n}\right\}+\frac{1}{\sqrt{n}}\\
&=\Pr\Big\{\sum_{i\in[n]}\imath_{q,q}(X_i;Y_i)\leq d\log M+\log(\sqrt{n})\Big\}+\frac{1}{\sqrt{n}}\label{step2},
\end{align}
where \eqref{usemany} follows from i) the change of measure technique which states that
\begin{align}
\Pr\{\imath_q(\barX^n;y^n)\geq t\}
&=\sum_{\barx^n}P_X^n(\barx^n)\bbo(\imath_q(\barx^n;y^n)\geq t)\leq \sum_{\barx^n} P_{X|Y}^q(\barx^n|y^n)\exp(-t)=\exp(-t),
\end{align}
and ii) the result in \cite[Eq. (37)]{scarlett2017mismatch} saying that $\mathbb{E}[\min\{1,J\}]\leq \Pr\{J>\frac{1}{\sqrt{n}}\}+\frac{1}{\sqrt{n}}$ for any $n\in\bbN$.

Now choose $M$ such that
\begin{align}
d\log M=nC+\sqrt{nV_\varepsilon}\Phi^{-1}(\varepsilon)-\frac{1}{2}\log n,
\end{align}
and let 
\begin{align}
\eta=\sqrt{\frac{d\log M}{2M^d}}=O\left(\frac{\sqrt{n}}{\exp(nC/2)}\right).
\end{align}
Thus, we have
\begin{align}
4n\exp(-2M^d\eta^2)
&=\frac{4n}{M^d}=
4\exp\bigg(-nC-\sqrt{nV_\varepsilon}\Phi^{-1}(\varepsilon)+\frac{3}{2}\log n\bigg)=O(\exp(-nC))\label{result1},
\end{align}
and
\begin{align}
\exp(n\eta c(q))
&=1+n\eta c(q)+o(n\eta c(q))
=1+O\left(\frac{n^{3/2}}{\exp(nC/2)}\right)\label{result2}.
\end{align}

Finally, applying the Berry-Esseen theorem to \eqref{step2}, we have that for any $q\in\calP_{\rm{ca}}$ and any $\varepsilon\in[0,1)$,
\begin{align}
\mathbb{E}[\min\{1,M^d\Pr\{\imath_{q}(\barX^n;Y^n)\geq \imath_{q}(X^n;Y^n)\}\}]
&\leq \varepsilon+O\left(\frac{1}{\sqrt{n}}\right)\label{result3}.
\end{align}

Combining \eqref{step1} and the results in \eqref{result1} to \eqref{result3}, we have that for $n$ sufficiently large,
\begin{align}
-d\log\delta^*(n,d,\varepsilon)
&\geq nC+\sqrt{nV}\Phi^{-1}(\varepsilon)-\frac{1}{2}\log n.
\end{align}

\subsubsection{Converse Proof}
We now proceed with the converse proof. Given any $\varepsilon\in[0,1)$, for any $\beta\in(0,\frac{1-\varepsilon}{2})$ and any $\kappa\in(0,1-\varepsilon-2d\beta)$, from Theorem \ref{fbl:converse}, we have
\begin{align}
-d\log \delta^*(n,d,\varepsilon)\leq -d\log \beta-\log\kappa+\sup_{\calA^n\in[0,1]^{nd}}\sup\bigg\{t\Big|\Pr\Big\{\sum_{i\in[n]}\imath_{\calA_i}(X_i;Y_i)\leq t\bigg\}\leq \varepsilon+2d\beta+\kappa\Big\}\label{step1:c}.
\end{align}
We first analyze the probability term in \eqref{step1:c}. Given any sequence of queries $\calA^n$, let
\begin{align}
C_{\calA^n}&:=\frac{1}{n}\sum_{i\in[n]}\mathbb{E}[\imath_{\calA_i}(X_i;Y_i)]\label{def:CA^n},\\
V_{\calA^n}&:=\frac{1}{n}\sum_{i\in[n]}\mathrm{Var}[\imath_{\calA_i}(X_i;Y_i)],\\
T_{\calA^n}&:=\frac{1}{n}\sum_{i\in[n]}\mathbb{E}[|\imath_{\calA_i}(X_i;Y_i)-\mathbb{E}[\imath_{\calA_i}(X_i;Y_i)]|^3],
\end{align}
Assume that there exists $V_->0$ such that $V_-\leq V_{\calA^n}$. Applying the Berry-Esseen theorem~\cite{berry1941accuracy,esseen1942liapounoff}, we have that
\begin{align}
\sup\bigg\{t \Big|\Pr\Big\{\sum_{i\in[n]}\imath_{\calA_i}(X_i;Y_i)\leq t\Big\}\leq \varepsilon+2d\beta+\kappa\bigg\}\leq nC_{\calA^n}+\sqrt{nV_{\calA^n}}\Phi^{-1}\bigg(\varepsilon+2d\beta+\kappa +\frac{6T_{\calA^n}}{\sqrt{nV_-^3}}\bigg)\label{step2:c}.
\end{align}
Let $\beta$ and $\kappa$ be chosen so that
\begin{align}
d\beta=\kappa=\frac{1}{\sqrt{n}}.
\end{align}

Using \eqref{step1:c} and \eqref{step2:c}, we have
\begin{align}
-d\log\delta^*(n,d,\varepsilon)
&\leq \log n+\sup_{\calA^n\in[0,1]^{nd}}\bigg(nC_{\calA^n}+\sqrt{nV_{\calA^n}}\Phi^{-1}\bigg(\varepsilon+\frac{2}{\sqrt{n}} +\frac{6T_{\calA^n}}{\sqrt{nV_-^3}}\bigg)\bigg)\label{touseinl2}.
\end{align}

For any sequence of queries $\calA^n$, we have
\begin{align}
C_{\calA^n}
&\leq \sup_{\calA\subseteq[0,1]^d}\mathbb{E}[\imath_{\calA}(X;Y)]=\sup_{p\in[0,1]}\mathbb{E}[\imath_{p}(X;Y)]=C\label{max:C1}.
\end{align}
Combining \eqref{touseinl2} and \eqref{max:C1}, when $n$ is sufficiently large, for any $\varepsilon\in[0,1)$,
\begin{align}
-d\log\delta^*(n,d,\varepsilon)
&\leq \log n+\sup_{\calA^n:|\calA_i|=q^*,~\forall i\in[n]}\bigg(nC_{\calA^n}+\sqrt{nV_{\calA^n}}\Phi^{-1}\bigg(\varepsilon+\frac{2}{\sqrt{n}} +\frac{6T_{\calA^n}}{\sqrt{nV_-^3}}\bigg)\bigg)\label{touseinl3}\\
&=\log n+nC+\sqrt{nV_\varepsilon}\Phi^{-1}\bigg(\varepsilon+\frac{2}{\sqrt{n}}+\frac{6T_{\calA^n}}{\sqrt{nV_-^3}}\bigg)\label{usedefveps}\\
&=nC+\sqrt{nV_\varepsilon}\Phi^{-1}(\varepsilon)+\log n+O(1)\label{taylor},
\end{align}
where \eqref{touseinl3} follows since i) for any $i\in[n]$, the maximum value of $\mathbb{E}[\imath_{\calA_i}[X_i;Y_i]]$ is achieved by any query $\calA_i$ with size $q^*$ which achieves the capacity $C$ and ii) when $n$ is sufficiently large, $nC_{\calA^n}=\frac{1}{n}\sum_{i\in[n]}\mathbb{E}[\imath_{\calA_i}[X_i;Y_i]]$ is the dominant term in the supremum, \eqref{usedefveps} follows from the definition of $V_\varepsilon$ in \eqref{def:veps} and \eqref{taylor} follows from the Taylor's expansion of $\Phi^{-1}(\cdot)$ (cf. \cite[Eq. (2.38)]{TanBook}) and the fact that $T_{\calA^n}$ is finite for discrete random variables $X^n$ and $Y^n$.

\subsection{Proof of Second-Order Achievable Asymptotics for Adaptive Querying (Theorem \ref{second:fbl:adaptive})}
\label{proof:second:fbl:adaptive}

\subsubsection{An Non-Asymptotic Achievability Bound}
In this subsection, we present an adaptive query procedure based on the variable length feedback code in \cite[Definition 1]{polyanskiy2011feedback} and analyze its non-asymptotic performance.

Let $\bX^\infty$ be a collection of $M^d$ random binary vectors $\{X^\infty(i_1,\ldots,i_d)\}_{(i_1,\ldots,i_d)\in[M]^d}$, each with infinite length and let $\bx^\infty$ denote a realization of $\bX^\infty$. Furthermore, let $Y^\infty$ be another random vector with infinite length where each element takes values in $\calY$ and let $y^\infty$ be a realization of $Y^\infty$. For any $\bw\in[0,1]^d$ and any $n\in\bbN$, given any sequence of queries $\calA^n=(\calA_1,\ldots,\calA_n)\in[0,1]^d$, define the following joint distribution of $(\bX^n,Y^n)$
\begin{align}
P_{\bX^nY^n}^{\calA^n,\bw}(\bx^n,y^n)
&=\prod_{t\in[n]}\Big(\prod_{(i_1,\ldots,i_d)\in[M]^d}\mathrm{Bern}_p(x_t(i_1,\ldots,i_d))\Big)P_{Y|X}^{\calA_t}(y_t|x_t(\bw))\label{def:pxyan}.
\end{align}
We can define $P_{\bX^\infty,Y^\infty}^{\calA^n,\bw}$ as a generalization of $P_{\bX^n,Y^n}^{\calA^n,\bw}$ with $n$ replaced by $\infty$. Since the channel is memoryless, such a generalization is reasonable.

Given any $(d,M)\in\bbN^2$, define a function $\Gamma:[M]^d\to [M^d]$ as follows: for any $(i_1,\ldots,i_d)\in[M]^d$,
\begin{align}
\Gamma(i_1,\ldots,i_d)=1+\sum_{j\in[d]}(i_j-1)M^{d-j}\label{def:Gamma}.
\end{align}
Note that the function $\Gamma(\cdot)$ is invertible. We denote $\Gamma^{-1}:[M^d]\to [M]^d$ the inverse function. Furthermore, given any $\lambda\in\bbR_+$ and any $m\in[M^d]$, define the stopping time
\begin{align}
\tau_m(\bx^\infty,y^\infty)&:=\inf\{n\in\bbN:~\imath_q(x^n(\Gamma^{-1}(m));y^n)\geq \lambda\}\label{def:taum}.
\end{align}
Our non-asymptotic bound states as follows.
\begin{theorem}
\label{fbl:ach:adaptive}
Given any $(d,M)\in\bbR_+\times\bbN$, for any $p\in[0,1]$ and $\lambda\in\bbR_+$, there exists an $(l,d,\frac{1}{M},\varepsilon)$-adaptive query procedure such that
\begin{align}
l&\leq \mathbb{E}[\tau_1(\bX^\infty,Y^\infty)],\\
\varepsilon&\leq(M^d-1)\Pr\{\tau_1(\bX^\infty,Y^\infty)\geq \tau_2(\bX^\infty,Y^\infty)\},
\end{align}
where the expectation and probability are calculated with respect to $P_{\bX^\infty,Y^\infty}^{\calA^\infty,\Gamma^{-1}(1)}$.
\end{theorem}

\begin{proof}[Proof of Theorem \ref{fbl:ach:adaptive}]
The proof of Theorem \ref{fbl:ach:adaptive} is inspired by \cite[Theorem 3]{polyanskiy2011feedback} and is largely similar to the proof for non-adaptive query procedures in Appendix \ref{proof:ach}. Thus, we only emphasize the differences here.

To prove Theorem \ref{fbl:ach:adaptive}, we analyze an adaptive query procedure based on the variable length feedback code in \cite{polyanskiy2011feedback}, which is stated as follows. Let $\bx=\{x^\infty(i_1,\ldots,i_d)\}_{(i_1,\ldots,i_d)\in[M]^d}$ be a sequence of $M^d$ binary codewords with infinite length. Then for any $n\in\bbN$ and any $(i_1,\ldots,i_d)\in[M]^d$, let $X^n(i_1,\ldots,i_d)$ be the first $n$ elements of $X^\infty(i_1,\ldots,i_d)$. Similarly to the proof of Theorem \ref{ach:fbl} in Appendix \ref{proof:ach}, we use the query $\calA_t$ as in \eqref{def:query:ddim} and apply the quantization function $\rmq(\cdot)$ in \eqref{def:qs} to generated quantized targets $\bw=(w_1,\ldots,w_d)$, i.e., $w_i=\rmq(s_i)$ for each $i\in[d]$ given any target variable $\bs=(s_1,\ldots,s_d)\in[0,1]^d$. The noiseless response to the query $\calA_t$ is then $X_t(\bw)$ and the noisy response $y_t$ is obtained by passing $x_t$ through the measurement-dependent channel $P_{Y|X}^{\calA_t}$.

The decoding process is summarized as follows, which includes the design of the stopping time and decoding function. Let $\lambda\in\bbR_+$ be a fixed threshold. Recall the definitions of $\Gamma(\cdot)$ in \eqref{def:Gamma} and $\tau_m(\bx^\infty,y^\infty)$ in \eqref{def:taum}. For any $(M,d)\in\bbN^2$, the stopping time is chosen as
\begin{align}
\tau^*(\bx^\infty,y^\infty):=\min_{m\in[M^d]}\tau_m(\bx^\infty,y^\infty).
\end{align}
The decoder outputs estimates $\hat{\bS}=(\hatS_1,\ldots,\hatS_d)$ via the following  two-stage decoding
\begin{enumerate}
\item the decoder first generates estimates $\hat{\bW}=(\hatW_1,\ldots,\hatW_d)$ as follows:
\begin{align}
\hat{\bW}=\Gamma^{-1}(\hatt),~\hatt=\max\{t\in[M^d]:\tau_j(\bx^\infty,y^\infty)=\tau^*(\bx^\infty,y^\infty)\},
\end{align}
\item the decoder produces estimates $\hat{\bS}=(\hatS_1,\ldots,\hatS_d)$ as
\begin{align}
\hatS_i=\frac{2\hatW_i-1}{2M},~i\in[d].
\end{align}
\end{enumerate}

Using the above adaptive query procedure, we have that the average stopping time satisfies
\begin{align}
\sup_{f_{\bS}\in\calF([0,1]^d)}\mathbb{E}[\tau^*(\bx^\infty,Y^\infty)]
&=\sup_{f_{\bS}\in\calF([0,1]^d)}\int_{\bs\in[0,1]^d}f_{\bS}(\bs)\mathbb{E}[\tau^*(\bx^\infty,Y^\infty)|\bS=\bs]\\
&=\sup_{P_{\bW}\in\calP([M]^d)}\sum_{\bw\in[M]^d}P_{\bW}(\bw)\mathbb{E}[\tau^*(\bx^\infty,Y^\infty)|\bW=\bw]\\
&\leq \sup_{P_{\bW}\in\calP([M]^d)}\sum_{\bw\in[M]^d}P_{\bW}(\bw)\mathbb{E}[\tau_{\Gamma(\bw)}(\bx^\infty,Y^\infty)|\bW=\bw]\label{delay},
\end{align}
and the excess-resolution probability with respect to the resolution $\delta=\frac{1}{M}$ satisfies
\begin{align}
\nn&\sup_{f_{\bS}\in\calF([0,1]^d)}\Pr\{\exists~i\in[d],~|\hatS_i-S_i|>\delta\}\\*
&\leq \sup_{P_{\bW}\in\calP([M]^d)}\Pr\{\hat{\bW}\neq\bW\}\\
&\leq \sup_{P_{\bW}\in\calP([M]^d)}\sum_{\bw\in[M]^d}P_\bW(\bw)\Pr\{\tau_{\Gamma(\bw)}(\bx^\infty,Y^\infty)\geq \tau^*(\bx^\infty,Y^\infty)\}\label{errorp}.
\end{align}
In the following, we will show that there exists binary codewords $\bx^{\infty}$ such that the results in \eqref{delay} and \eqref{errorp} are upper bounded by the desired bounds in Theorem \ref{ach:fbl}.

Let $\bX^{\infty}:=\{X^\infty(i_1,\ldots,i_d)\}_{(i_1,\ldots,i_d)\in[M]^d}$ be a sequence of $M^d$ binary codewords with infinite length where each codeword is generated i.i.d. from the Bernoulli distribution $P_X$ with parameter $p\in(0,1)$. For any $\bw\in[0,1]^d$ and any $n\in\bbN$, using the above adaptive query procedure, the joint distribution of $(\bX^n,Y^n)$ is $P_{\bX^n,Y^n}^{\calA^n,\bw}(\bx^n,y^n)$ as defined in \eqref{def:pxyan}. 

For any $P_{\bW}\in\calP([M]^d)$, we have
\begin{align}
\mathbb{E}_{\bX^\infty}[\tau^*(\bX^\infty,Y^\infty)]
&=\sum_{\bw\in[M]^d}P_{\bW}(\bw)\mathbb{E}_{P_{\bX^\infty,\bY^\infty}^{\calA^n,\bw}}[\tau^*(\bX^\infty,Y^\infty)]\\
&\leq \sum_{\bw\in[M]^d}P_{\bW}(\bw)\mathbb{E}_{P_{\bX^\infty,\bY^\infty}^{\calA^n,\bw}}[\tau_{\Gamma(\bw)}(\bX^\infty,Y^\infty)]\\
&=\sum_{\bw\in[M]^d}P_{\bW}(\bw)\mathbb{E}_{P_{\bX^\infty,\bY^\infty}^{\calA^n,\Gamma^{-1}(1)}}[\tau_1(\bX^\infty,Y^\infty)]\label{symmetry}\\
&=\mathbb{E}_{P_{\bX^\infty,\bY^\infty}^{\calA^n,\Gamma^{-1}(1)}}[\tau_1(\bX^\infty,Y^\infty)],
\end{align}
where \eqref{symmetry} follows since for each $\bw\in[M]^d$, from the definition of $\tau_{\cdot}(\cdot)$ in \eqref{def:taum}, 
\begin{align}
\mathbb{E}_{P_{\bX^\infty,\bY^\infty}^{\calA^n,\bw}}[\tau_{\Gamma(\bw)}(\bX^\infty,Y^\infty)]=\mathbb{E}_{P_{\bX^\infty,\bY^\infty}^{\calA^n,\mathrm{ones}(d)}}[\tau_{\Gamma(\mathrm{ones}(d))}(\bX^\infty,Y^\infty)]=\mathbb{E}_{P_{\bX^\infty,\bY^\infty}^{\calA^n,\Gamma^{-1}(1)}}[\tau_1(\bX^\infty,Y^\infty)],
\end{align}
and we use $\mathrm{ones}(d)$ to denote the all one vector with length $d$.

Similarly, we have
\begin{align}
\mathbb{E}_{\bX^{\infty}}[\Pr[\hat{\bW}\neq \bW]]
&\leq \sum_{\bw\in[M]^d}P_{\bW}(\bw)\Pr_{P_{\bX^\infty,\bY^\infty}^{\calA^n,\bw}}\{\tau_{\Gamma(\bw)}(\bX^\infty,Y^\infty)\geq \tau^*(\bX^\infty,Y^\infty)\}\\
&=\sum_{\bw\in[M]^d}P_{\bW}(\bw)\Pr_{P_{\bX^\infty,\bY^\infty}^{\calA^n,\Gamma^{-1}(1)}}\{\tau_1(\bX^\infty,Y^\infty)\geq \tau^*(\bX^\infty,Y^\infty)\}\label{usesymmetry11}\\
&=\Pr_{P_{\bX^\infty,\bY^\infty}^{\calA^n,\Gamma^{-1}(1)}}\{\tau_1(\bX^\infty,Y^\infty)\geq \tau^*(\bX^\infty,Y^\infty)\}\\
&\leq (M^d-1)\Pr_{P_{\bX^\infty,\bY^\infty}^{\calA^n,\Gamma^{-1}(1)}}\{\tau_1(\bX^\infty,Y^\infty)\geq \tau_2(\bX^\infty,Y^\infty)\}\label{usesymmetry22},
\end{align}
where \eqref{usesymmetry11} follows from the symmetry which implies that $\Pr\{\tau_{\Gamma(\bw)}(\bX^\infty,Y^\infty)\geq \tau^*(\bX^\infty,Y^\infty)\}=\Pr\{\tau_1(\bX^\infty,Y^\infty)\geq \tau^*(\bX^\infty,Y^\infty)\}$ for any $\bw\in[M]^d$ and \eqref{usesymmetry22} follows from the union bound and the symmetry similar to \eqref{usesymmetry11}.

The proof of Theorem \ref{fbl:ach:adaptive} is completed by using the simple fact that $\mathbb{E}[X]\leq a$ implies that there exists $x\leq a$ for any random variable $X$ and constant $a\in\bbR$.
\end{proof}

\subsubsection{Proof of Achievable Second-Order Asymptotics}
\label{proof:second:ach:haha}
The proof of second-order asymptotics for adaptive querying proceeds similarly as \cite{polyanskiy2011feedback} and we only highlight the differences here. Let $q^*\in\calP_{\rm{ca}}$ be a capacity-achieving parameter for measurement-dependent channels $\{P_{Y|X}^q\}_{q\in[0,1]}$. From Theorem \ref{fbl:ach:adaptive}, we have that there exists an $(l,d,\frac{1}{M},\varepsilon)$-adaptive query procedure such that 
\begin{align}
l&\leq \mathbb{E}[\tau_1(\bX^\infty,Y^\infty)],\\
\varepsilon&\leq(M^d-1)\Pr\{\tau_1(\bX^\infty,Y^\infty)\geq \tau_2(\bX^\infty,Y^\infty)\}\label{touppe}.
\end{align}
Unless otherwise stated, the expectation and probability are calculated with respect to slight generalization of the joint distribution $P_{\bX^n Y^n}^{\calA^n,\Gamma^{-1}(1)}$ in \eqref{def:pxyan}.

For subsequent analyses, let $P_X$ be the Bernoulli distribution with parameter $q^*$ and let $\tilP_{XY}$ be the following joint distribution
\begin{align}
\tilP_{XY}(x,y)
&:=\sum_{\barx_1,\ldots,\barx^{M^d-1}}P_X(x)\bigg(\prod_{j\in[M^d-1]}P_X(\barx_j)\bigg)P_{Y|X}^{\frac{x+\sum_{j\in[M^d-1]}\barx_j}{M^d}}(y|x).
\end{align} 
Note that $\tilP_{XY}$ is the marginal distribution of $(X_i(\Gamma^{-1}(1)),Y_i)$ for each $i\in[n]$ induced from $P_{\bX^n\bY^n}^{\calA^n,\Gamma^{-1}(1)}$ under our query procedure.

Furthermore, define the ``mismatched'' version of the capacity.
\begin{align}
C_1&:=\mathbb{E}_{\tilP_{XY}}[\imath_{q^*}(X;Y)]\label{def:C1}.
\end{align}
Finally, for each $n\in\bbN$, let
\begin{align}
U_n&:=\imath_{q^*}(X^n;Y^n)=\sum_{i\in[n]}\imath_{q^*,q^*}(X_i;Y_i).
\end{align}
It can be easily verified that $\{U_n-nC_1\}_{n\in\bbN}$ is a martingale and for each $n\in\bbN$, $\mathbb{E}[U_n-nC_1]=0$. The optional stopping theorem~\cite[Theorem 10.10]{williams1991probability} implies that
\begin{align}
0&=\mathbb{E}[U_{\tau_1(\bX^\infty,\bY^\infty)}-C_1\tau_1(\bX^\infty,\bY^\infty)]\\
&\leq \lambda+a_0-C_1\mathbb{E}[\tau_1(\bX^\infty,\bY^\infty)]]\label{uppl2},
\end{align}
where $a_0$ is a uniform upper bound on the information density $U_1$.
Thus,
\begin{align}
\mathbb{E}[\tau_1(\bX^\infty,\bY^\infty)]
&\leq \frac{\lambda+a_0}{C_1}\label{uppl3}.
\end{align}
We then focus on upper bounding \eqref{touppe}. From \eqref{uppl3}, we have that
\begin{align}
\Pr\{\tau_1(\bX^\infty,\bY^\infty)<\infty\}=1,
\end{align}
since otherwise the expectation value of $\tau_1(\bX^\infty,\bY^\infty)$ would be infinity.

Recall the definition of the typical set $\calT(\cdot)$ in \eqref{def:typical}. For any $\eta\in\bbR_+$, we have
\begin{align}
\nn&\Pr\{\tau_1(\bX^\infty,Y^\infty)\geq \tau_2(\bX^\infty,Y^\infty)\}\\*
&\leq \Pr\{\tau_2(\bX^\infty,Y^\infty)<\infty\}\\
&=\sum_{t\in\bbN}\bbo(t<\infty)\Pr\{\tau_2(\bX^\infty,Y^\infty)=t\}\\
&=\sum_{t\in\bbN}\bbo(t<\infty)\Big\{\Pr\{\tau_2(\bX^\infty,Y^\infty)=t,\bX^t\in\calT^t(M,d,q^*,\eta)\}+\Pr\{\bX^t\notin\calT^t(M,d,q^*,\eta)\}\Big\}\\
&\leq\sum_{t\in\bbN}\bbo(t<\infty)\Big\{\exp(t\eta c(p))\Pr_{P_{\bX^\infty,Y^\infty}^{q^*,\Gamma^{-1}(1)}}\{\tau_2(\bX^\infty,Y^\infty)=t\}+4t\exp(-2M^d\eta^2)\Big\}\label{usetypicalha},
\end{align} 
where \eqref{usetypicalha} follows from \eqref{fromassumption} and the upper bound on the probability of atypicality similar to \eqref{useatypical} and in \eqref{usetypicalha}, we use the change-of-measure technique and the distribution $P_{\bX^\infty,Y^\infty}^{q^*,\Gamma^{-1}(1)}$ is a generalization of $P_{\bX Y^n}^{p,\bw}(\cdot)$ in \eqref{altdis} to an infinite length.

Given any $l'\in\bbR_+$, let $(\lambda,\eta,M)\in\bbR_+^2\times\bbN$ be chosen so that
\begin{align}
\lambda&=l'C_1-a_0,\label{herechooselambdahehe}\\
d\log M&=\lambda-\log l'\label{herechooseMhehe},\\
\eta&:=\sqrt{\frac{d\log M}{2M^d}}=O\left(\frac{\sqrt{l'}}{\exp(l'C_1/2)}\right).
\end{align}
Then from \eqref{uppl3}, we have
\begin{align}
\mathbb{E}[\tau_1(\bX^\infty,Y^\infty)]\leq l'.
\end{align}
Furthermore, similarly to \cite[Section D]{polyanskiy2011feedback}, we have
\begin{align}
\nn&\Pr\{\tau_1(\bX^\infty,Y^\infty)\geq \tau_2(\bX^\infty,Y^\infty)\}\\*
&=\sum_{t\in\bbN}\bbo(t<\infty)\left(1+O\left(l'^{\frac{3}{2}}\exp\left(-\frac{l'C_1}{2}\right)\right)\right)\Pr_{P_{\bX^\infty,Y^\infty}^{q^*,\Gamma^{-1}(1)}}\{\tau_2(\bX^\infty,Y^\infty)=t\}\\
&=\left(1+O\left(l'^{\frac{3}{2}}\exp\left(-\frac{l'C_1}{2}\right)\right)\right)\lim_{t\to\infty}\Pr_{P_{\bX^\infty,Y^\infty}^{q^*,\Gamma^{-1}(1)}}\{\tau_2(\bX^\infty,Y^\infty)<t\}\\
&=\left(1+O\left(l'^{\frac{3}{2}}\exp\left(-\frac{l'C_1}{2}\right)\right)\right)\lim_{t\to\infty}\mathbb{E}_{P_{\bX^\infty,Y^\infty}^{q^*,\Gamma^{-1}(1)}}\big[\exp(-U_t)\bbo(\tau_1(\bX^\infty,Y^\infty)<t)\big]\label{applycof}\\
&=\left(1+O\left(l'^{\frac{3}{2}}\exp\left(-\frac{l'C_1}{2}\right)\right)\right)\mathbb{E}_{P_{\bX^\infty,Y^\infty}^{q^*,\Gamma^{-1}(1)}}\big[\exp(-U_{\tau_1(\bX^\infty,Y^\infty)})\bbo(\tau_1(\bX^\infty,Y^\infty)<\infty)\big]\label{followfblwfb}\\
&\leq \left(1+O\left(l'^{\frac{3}{2}}\exp\left(-\frac{l'C_1}{2}\right)\right)\right)\exp(-\lambda)\label{def:tau1},
\end{align}
where \eqref{applycof} follows from the change-of-measure technique, \eqref{def:tau1} follows from the definition of $\tau_1(\bX^\infty,Y^\infty)$ in \eqref{def:taum} and \eqref{followfblwfb} follows similarly as \cite[Eq. (113) to Eq. (117)]{polyanskiy2011feedback} and the details are as follows:
\begin{align}
\nn&\lim_{t\to\infty}\mathbb{E}_{P_{\bX^\infty,Y^\infty}^{\mathrm{md},1}}\big[\exp(-U_t)\bbo(\tau_1(\bX^\infty,Y^\infty)<t)\big]\\*
&=\lim_{t\to\infty}\mathbb{E}_{P_{\bX^\infty,Y^\infty}^{\mathrm{md},1}}\big[\exp(-U_t)\bbo(\tau_t(\bX^\infty,Y^\infty)<t)\big]\label{def:taut}\\
&=\lim_{t\to\infty}\mathbb{E}_{P_{\bX^\infty,Y^\infty}^{\mathrm{md},1}}\big[\exp(-U_{\tau_t(\bX^\infty,Y^\infty)})\bbo(\tau_t(\bX^\infty,Y^\infty)<t)\big]\label{useost4mar}\\
&=\mathbb{E}_{P_{\bX^\infty,Y^\infty}^{\mathrm{md},1}}\big[\lim_{t\to\infty}\exp(-U_{\tau_t(\bX^\infty,Y^\infty)})\bbo(\tau_t(\bX^\infty,Y^\infty)<t)\big]\\
&=\mathbb{E}_{P_{\bX^\infty,Y^\infty}^{\mathrm{md},1}}\big[\exp(-U_{\tau_1(\bX^\infty,Y^\infty)})\bbo(\tau_1(\bX^\infty,Y^\infty)<\infty)\big]
\end{align}
where in \eqref{def:taut}, we define $\tau_t(\bX^\infty,Y^\infty):=\min\{\tau_1(\bX^\infty,Y^\infty),t\}$, \eqref{useost4mar} follows by applying the optional stopping theorem~\cite[Theorem 10.10]{williams1991probability} to the martingale $\exp(-U_t)$ and the stopping time $\tau_t(\bX^\infty,Y^\infty)$.

Thus, using \eqref{touppe}, for $l'$ sufficient large, the excess-resolution probability satisfies
\begin{align}
\varepsilon
&\leq (M^d-1)\Pr_{P_{\bX^\infty,Y^\infty}^1}\{\tau_1(\bX^\infty,Y^\infty)\geq \tau_2(\bX^\infty,Y^\infty)\}\\*
&\leq \left(1+O\left(l'^{\frac{3}{2}}\exp\left(-\frac{l'C_1}{2}\right)\right)\right)M^d\exp(-\lambda)\\
&=\left(1+O\left(l'^{\frac{3}{2}}\exp\left(-\frac{l'C_1}{2}\right)\right)\right)\frac{1}{l'}\label{usechooseMhehe},
\end{align}
where \eqref{usechooseMhehe} follows from the choice of $M$ in \eqref{herechooseMhehe}.

Recall the definition of the ``capacity'' $C$ of measurement-dependent channels in \eqref{def:capacity}. Using the definition of $C_1$ in \eqref{def:C1}, we have that
\begin{align}
C_1
&=\mathbb{E}_{\barP_{XY}}[\imath_{q^*}(X;Y)]\\
&=\mathbb{E}_{P_{XY}}\bigg[\frac{\barP_{XY}(X,Y)}{P_{XY}(X,Y)}\imath_{q^*}(X;Y)\bigg]\\
&\leq \exp(\eta c(p))\mathbb{E}_{P_{XY}}[\imath_{q^*}(X;Y)]+2\exp(-2M^d\eta^2)\label{changeofmeasureagain2}\\
&=\exp(\eta c(p))C+2\exp(-2M^d\eta^2).
\end{align}
where \eqref{changeofmeasureagain2} follows from the change-of-measure technique and the result in \eqref{fromassumption}. Given the choice of $M$ and $\eta$, we have
\begin{align}
C_1=C+O(l'\exp(-l')).
\end{align}

Thus, till now, we have proved that the above adaptive query procedure is an $(l',d,\exp(-\frac{l'C_1-\log l'-a_0}{d}),\frac{1}{l'})$-adaptive query procedure for sufficiently large $l'$ with proper choice of the parameters $(M,p,\lambda)$. For any $\varepsilon\in[0,1)$, consider the following query procedure with the drop out strategy: with probability $\varepsilon$, we do not pose any query and with the remaining probability, we use the above-constructed $(l',d,\exp(-\frac{l'C_1-\log l'-a_0}{d}),\frac{1}{l'})$-adaptive query procedure. For $l'$ sufficiently large, it is easy to verify that the drop out adaptive query procedure is an $(l,d,\varepsilon',\delta)$-adaptive query procedure where
\begin{align}
l&= (1-\varepsilon)l',\\
\varepsilon'
&=\varepsilon+\frac{1-\varepsilon}{l'}\approx \varepsilon,\\
-d\log \delta
&=l'C_1-\log l'-a_0=l' C+O(\log l')=\frac{Cl}{1-\varepsilon}+O(\log l).
\end{align}

\section*{Acknowledgments}
The authors acknowledge three anonymous reviewers for many helpful comments and suggestions, which significantly improve the quality of the current manuscript.

\bibliographystyle{IEEEtran}
\bibliography{IEEEfull_lin}

\end{document}